\documentclass[sigconf,nonacm]{acmart} 

\usepackage[dvipsnames]{xcolor}

\newcommand{\RR}{\mathbb{R}}
\usepackage{tabularx}

\newcommand{\mcE}{\mathcal{E}}

\newcommand{\mcP}{\mathcal{P}}

\newcommand{\NN}{\mathbb{N}}

\newcommand{\rk}{\mathrm{rk}}

\newcommand{\mbfv}{\mathbf{v}}
\newcommand{\mbfe}{\mathbf{e}}
\newcommand{\mbfa}{\mathbf{a}}
\newcommand{\mbfx}{\mathbf{x}}
\newcommand{\mbfp}{\mathbf{p}}
\newcommand{\mbfs}{\mathbf{s}}

\binoppenalty=\maxdimen
\relpenalty=\maxdimen

\usepackage[capitalize]{cleveref}

\newtheorem{conjecture}{Conjecture}
\newtheorem{ex}{Example}
\newtheorem{defn}{Definition}
\newtheorem{prop}{Proposition}

\usepackage{tikz}
\usepackage{tikz-3dplot}
\usepackage{tkz-euclide}
\usepackage{amsmath}
\usepackage{pgfplots}
\usepackage{pgfmath}
\pgfplotsset{compat=1.18}
\usepackage{thm-restate}
\usepackage{balance} % for balancing columns on the final page
\usepackage{xspace}

\title[Real Preferences]{Real Preferences under Arbitrary Norms}

\author{Joshua Zeitlin}
\affiliation{
	\institution{University of Cambridge}
	\country{United Kingdom}}
\email{jdz28@cam.ac.uk}

\author{Corinna Coupette}
\affiliation{
	\institution{Aalto University}
	\country{Finland}}
\email{corinna.coupette@aalto.fi}

\begin{abstract}
Whether the goal is to reach decisions among multiple agents, 
ensure that AI systems are aligned with human preferences, 
or design better recommender systems,
the problem of translating between (ordinal) rankings and (numerical) utilities arises naturally in many contexts. 
This task is commonly approached by computing \emph{embeddings},  
which represent both the agents doing the ranking (\emph{voters}) and the items to be ranked (\emph{alternatives}) in a shared metric space. 
Here, ordinal preferences are translated into relationships between pairwise distances. 
Prior work has established that any collection of rankings with $n$ voters and $m$ alternatives (\emph{preference profile}) can be embedded into $d$-dimensional Euclidean space for $d \geq \min\{n,m-1\}$ under the Euclidean norm and the Manhattan norm. 
We show that this holds for \emph{all $p$-norms} and establish that any \emph{pair} of rankings can be embedded into $\mathbb{R}^2$ under \emph{arbitrary norms}, 
significantly expanding the reach of spatial preference models. 

\end{abstract}


\begin{document}

\pagestyle{fancy}
\fancyhead{}

\maketitle 

% !TeX spellcheck = en_US
% !TEX root = ../main.tex 

\section{Introduction}\label{sec:introduction}

%Accurate modelling of preferences is critical to the safe and
%effcient deployment of AI in multiagent systems [Christian,
%2020]. Indeed, if the preference models used by AI agents
%are not suffciently expressive to be capable of representing
%human preferences, they will be at least partially “mistaken”
%about what humans want and, consequently, may take actions
%that cause harm. For example, a recommender system built
%on inaccurate preference models may make recommendations
%that bias perceptions of politics [Huszar´ et al., 2021], contribute to low self esteem [Faelens et al., 2021], or encourage
%unsafe medical interventions [Johnson et al., 2021].

From selecting winners in real-world elections 
to reaching decisions in multiagent systems, 
designing powerful recommendation engines, 
and improving language models with human feedback: 
The problem of evaluating several alternatives based on information about preferences collected from multiple parties lies at the heart of diverse practical applications. 
In many settings, 
the preference information available is \emph{ordinal}---%
e.g., 
because rankings and pairwise comparisons can be intuitively understood by humans and communicated efficiently between agents.
At the same time, 
\emph{metric} representations of preferences tend to allow more expressive evaluations and more efficient computations. 
As a result, identifying ``good'' mappings from ordinal to metric preferences constitutes a key practical challenge. 

In economics and political science, this challenge is commonly approached via the widely popular \emph{theory of spatial preferences}, 
which represents agents (e.g., voters) and alternatives (e.g., candidates or policy outcomes) as points in Euclidean space \citep{stokes1963spatial,enelow1984spatial}. 
Here, the underlying norm is often assumed to be Euclidean as well; 
individual utilities and collective social costs are treated as functions of Euclidean distance \citep{kalandrakis2010rationalizable,eguia2013spatial}; 
and the dimensionality $d$ of the target space is typically low (i.e., $1 \leq d \leq 3$). 
The spatial model of preferences rests on solid theoretical foundations \citep{enelow1984spatial,eguia2011foundations,chambers2020spherical},  
and latent-space models could be viewed as implicitly extending these foundations beyond the political realm. 
%and weighted preferences can be used to capture the non-separability of issues \citep{davis1970expository,hinich1997analytical,binding2023non}.
However, while much \emph{methodological} effort has focused on estimating voters' ideal points from real-world preference data \citep{ladha1991spatial,clinton2004statistical,kim2018estimating,carroll2013structure,luque2025operationalizing},
a growing body of work indicates that low-dimensional Euclidean space under the Euclidean norm provides a poor fit---%
both \emph{conceptually} \citep{milyo2000logical,eguia2009utility,humphreys2010spatial} 
and \emph{empirically} \citep{pardos2010systemic,ye2011evaluating,henry2013euclidean,stoetzer2015multidimensional,yu2021spatial}---%
even in political applications. 

\textbf{Motivation and Related Work.}\quad
Given the importance of translating between ordinal and metric representations of preference information 
in recommender systems, multiagent decision-making, 
and social-choice applications, 
it is natural to ask under which conditions a given \emph{preference profile} 
(i.e., a list of rankings over alternatives) permits what we call a \emph{rank-preserving embedding}: 
a mapping of voters and alternatives into a normed real vector space that encodes the available ordinal preference information in the relationships between the pairwise distances. 

Toward answering this question, 
researchers in computational social choice have
identified \emph{necessary conditions} in terms of forbidden substructures \citep{peters2017recognising,thorburn2023error}, 
and they have analyzed the computational complexity of recognizing when a preference profile admits a rank-preserving embedding into Euclidean space of a specified dimension 
\citep{knoblauch2010recognizing,peters2017recognising,escoffier2023algorithmic}.
They have also formulated \emph{sufficient conditions} regarding the dimensionality of the embedding space for the special cases of the Euclidean norm~\citep{bogomolnaia2007euclidean}
and the Manhattan norm \citep{chen2022manhattan}. 
However, to our knowledge, nothing is known about the rank embeddability of preference profiles under \emph{arbitrary norms}. 

The dearth of rank-embeddability results under norms that are neither Euclidean nor Manhattan 
stands in stark contrast to the growing interest in concepts requiring the joint consideration of \emph{all metric spaces}, such as \emph{metric distortion} 
\citep[see][]{anshelevich2018approximating,anagnostides2022metric,caragiannis2022metric,ebadian2024metric}. 
Moreover, non-Euclidean norms have proved to be relevant across a variety of contexts---%
from \emph{voting} \citep{peters1993generalized,humphreys2010spatial,eckert2010equity,chen2022manhattan,shin2025l1} and \emph{facility location} \citep{larson1983facility,feigenbaum2017approximately,lee2025facility} 
to \emph{algorithmic data analysis}
\citep{cohen2019inapproximability,cohen2022johnson,cohen2023breaching,abbasi2023parameterized} and \emph{recommender systems} \citep{rudin2009pnorm}. 
Overall, a better understanding of the restrictions imposed on rank embeddability by more general norms appears desirable. 
Our work makes significant progress toward this goal.

\textbf{Formal Setting and Intuition.}\quad
We operate in a setting with $m$ alternatives $A = \{a_1,\dots,a_m\}$ (also known as candidates) and $n$ voters $V = \{1,\dots,n\}$. 
Given a set of alternatives $A$ and a voter $i$, 
a \emph{preference} of $i$ is a weak order $\succsim_i$ on $A$, 
where $x \succsim_i y$ indicates that $i$ weakly prefers alternative $x$ to alternative $y$. 
A preference is \emph{strict} if $x \sim_i y$ implies $x = y$, 
and it is \emph{complete} if we have $x \succsim_i y$ or $y \succsim_i x$ for all $x \neq y$.
We will generally assume that preferences are strict and complete. 

Given alternatives $A$ and voters $V$ with preferences over $A$, 
the list $(\succsim_i)^n_{i=1}$ is said to be a \emph{preference profile} $\mcP_{A,V}$,
where we drop the subscript of $\mcP$ when $A$ and $V$ are clear from context. 
We are interested in characterizing when a preference profile can be faithfully represented in a given metric space, 
focusing on real metric spaces equipped with a \emph{norm}.
\begin{defn}[Norm]\label{def:norm}
	A \emph{norm} on $\RR^d$ is a function $\|\cdot\|\colon\RR^d\to \RR^d$ such that 
	$\|\mbfx\|\geq 0$ (non-negativity), 
	$\|\mbfx+\mathbf{y}\|\leq \|\mbfx\|+\|\mathbf{y}\|$ (subadditivity), 
	$\|\mbfx\|=0$ if and only if $x=0$ (positive definiteness), 
	and $\|c\mbfx\|=|c|\|\mbfx\|$ for any $c\in \RR$ (absolute homogeneity)\;.
\end{defn}
Some of our results focus specifically on \emph{$p$-norms}.
\begin{defn}[$p$-Norm]\label{def:pnorm}
	Given a real number $p\geq 1$, the $p$-norm of a vector $\mbfx\in\RR^d$ is given by 
	$\|\mbfx\|_p \coloneq \left(\sum_{i=1}^d |x_i|^p\right)^{\frac{1}{p}}$.
\end{defn}

For real metric spaces with \emph{arbitrary} norms, 
we introduce a flexible notion of embeddability.
\begin{defn}[Rank-preserving embedding, rank embeddability]\label{def:embeddability}
	Given a preference profile $\mcP$, a dimension $d$, and a normed real vector space $(\RR^d,\|\cdot\|)$, 
	an assignment of coordinates $\mbfa_j\in \RR^d$ to alternatives $a_j\in A$ and coordinates $\mbfv_i\in \RR^d$ to voters $i\in V$ constitutes 
	a \emph{rank-preserving embedding} of $\mcP$ into $(\RR^d,\|\cdot\|)$ 
    if 
    \begin{align*}
    	a_j\succsim_ia_k  \iff \|\mbfv_i-\mbfa_j\|\leq\|\mbfv_i-\mbfa_k\|\;.
    \end{align*}
    If~there exists a rank-preserving embedding of $\mcP$ into $(\RR^d,\|\cdot\|)$, $\mcP$ is said to \emph{rank-embed} into $(\RR^d,\|\cdot\|)$.
\end{defn}

This definition generalizes the notions of \emph{$d$-Euclidean} and \emph{$d$-Manhattan} embeddings introduced by \citet{bogomolnaia2007euclidean} and \citet{chen2022manhattan}, respectively, to arbitrary norms, 
and it can be further generalized to arbitrary metric spaces.
Intuitively, \cref{def:embeddability} requires that a coordinate assignment translates higher ranks in preference orderings to smaller distances in the embedding space. 
Since any two voters with identical preferences can be assigned identical coordinates, 
this implies that the rank embeddability of $\mcP$ into a given metric space only depends on the \emph{distinct preferences} (i.e., \emph{equivalence classes}) of voters. 
Therefore, in the following, we assume w.l.o.g.\ that the preference profiles we study are \emph{irreducible}, 
that is, no two voters $i$ and $j$ have identical preferences. 
Consequently, $n$ can be interpreted as the number of distinct \emph{voter types}. 

\textbf{Main Contributions and Techniques.}\quad 
In their pioneering work, 
\citet{bogomolnaia2007euclidean} proved that any preference profile with $n$ voters and $m$ alternatives rank-embeds into $(\RR^d,\|\cdot\|_2)$ if $d\geq \min\{n,m-1\}$.
\citet{chen2022manhattan} extended this result to $(\RR^d,\|\cdot\|_1)$ 
using a very different approach, 
and they also treat the infinity norm using yet another construction.  
Leveraging our norm-independent notion of \emph{rank embeddability} (\cref{def:embeddability}), 
in \cref{sec:thm1}, 
we generalize these results to \emph{arbitrary $p$-norms} for $p>1$, 
streamlining the handling of the cases $p = \infty$ and $\infty > p > 1$ as a side effect. 
Separating the cases $d \geq n$ and $d\geq m-1$, 
our proofs are constructive and geometric, 
isolating only those features of the embedding classes introduced by \citet{bogomolnaia2007euclidean} that are strictly necessary to ensure rank preservation. 
Analyzing the asymptotic behavior of our main construction as $p \searrow 1$, 
we further elucidate why the $1$-norm requires a fundamentally different construction. 

Following our investigation of $p$-norms, 
we explore how our results extend to arbitrary norms in low dimensions.  
In \cref{sec:thm2}, we demonstrate that any preference profile with two voters rank-embeds into $(\RR^2,\|\cdot\|)$ for \emph{any} norm~$\|\cdot\|$ (including norms that are not $p$-norms, such as $2\|\cdot\|_1+5\|\cdot\|_2$). 
Our proof relies on a geometric construction that, 
while not fully explicit, 
may still be of independent interest in the context of facility-location problems such as those studied by \citet{larson1983facility} or \citet{chan2021mechanism}. 
Unlike \citet{bogomolnaia2007euclidean}, we do not need to resort to arithmetic calculations,  
which allows us to provide a much cleaner framework for reasoning about rank embeddability. 

Our results demonstrate the flexibility with which the theory of spatial preferences can translate between ordinal preferences and numerical utilities,  
and they provide a theoretical basis for systematically exploring the potential of non-standard norms in practical applications. 
As elaborated in \cref{sec:discussion}, 
our results further highlight where additional research is needed to understand the prerequisites and implications of spatial-preference theory.
They also constitute a strong foundation for tackling the rank-embeddability question under arbitrary norms in higher dimensions, which we formalize in \cref{sec:discussion} as \cref{conj}.

% !TeX spellcheck = en_US
% !TEX root = ../main.tex 
\section{Rank-Preserving Embeddings under \texorpdfstring{$p$}{TEXT}-Norms}\label{sec:thm1}

In this section, we generalize the rank-embeddability results originally proved for the 1-norm and the 2-norm to arbitrary $p$-norms, 
considerably extending the set of scenarios to which the theory of spatial preferences can be applied. 

\begin{restatable*}[Rank embeddability under {$p$-norms}]{thm}{lpnorms}\label{thm1}
	Given $m$ alternatives $A$ and $n$ voters $V$ with preferences over these alternatives, 
	a preference profile~$\mcP_{A,V}$ rank-embeds into $(\RR^d,\|\cdot\|_p)$, for all $1\leq p\leq \infty$, if $d\geq \min\{n,m-1\}$.
\end{restatable*}

Our proof proceeds in two steps, 
establishing rank embeddability into $(\RR^d, \|\cdot\|_p)$ first for $d \geq n$ (\Cref{prop:prop1}) and then for $d \geq m-1$ (\Cref{prop:prop2}).

\subsection{Required Dimensionality Depending on the Number of Voters}

To establish rank embeddability as a function of $n$, 
we leverage a family of embeddings introduced by \citet{bogomolnaia2007euclidean}. 
We refer to embeddings from this family as \emph{alternative-rank embeddings}---%
because the coordinates of the alternatives directly reflect how they are ranked by the voters. 
These embeddings are defined as follows.

\begin{defn}[Alternative-Rank Embedding {[}AR Embedding{]}]\label{def:alternative-rank-embedding}
	Given a preference profile $\mcP$ with $m$ alternatives and $n$ voters as well as a constant $c>0$, 
	an \emph{alternative-rank embedding} $\mcE_c$
	assigns coordinates $\mbfv_i=c\cdot\mbfe_i$ to each voter and coordinates $\mbfa_j=(-\rk_{i}j\mid i\in [n])$ to each alternative,
	 where $\mbfe_i$ is the $i^{th}$ standard basis vector
	 and $\rk_{i}j$ is the rank of alternative $j$ in the preference ordering of voter~$i$. 
\end{defn}

\citet{bogomolnaia2007euclidean} show that for sufficiently large $c$, 
any AR embedding $\mcE_c$ into $\RR^n$ is rank-preserving for any preference profile $\mcP$ with $n$ voters under the Euclidean norm.
The following example provides some intuition for this finding.

\begin{ex}\label{ex1}
	Given a setting with five alternatives, $a_1$ to $a_5$, 
	and~two voters, $v_1$ and $v_2$,
	consider the preference profile $\mcP = (\succsim_1,\succsim_2)$, with
	\begin{align*}
		\succsim_1 :~a_1\succ a_2\succ a_3\succ a_4\succ a_5&~\text{~and~}~
		\succsim_2 :~ a_2\succ a_4\succ a_5\succ a_1\succ a_3\;.
	\end{align*}
	\Cref{fig:ar-embedding} depicts a rank-preserving AR embedding of~~$\mcP$ using $c = 10$.
\end{ex}

\begin{figure}[t]
	\centering
	% !TEX root = ../main.tex 
\begin{tikzpicture}[scale=0.475]
	\newcommand{\figonecolorone}{MidnightBlue}
	\newcommand{\figonecolortwo}{OliveGreen}
	\newcommand{\figonecolorthree}{black}
	% Optional: Grid
	\draw[step=2cm,gray!50,very thin] (-6,-6) grid (11,11);
	
	\draw[-,thick,color=\figonecolorone] (10,0) -- (-1,-4);
	\draw[-,thick,color=\figonecolorone!80] (10,0) -- (-2,-1);
	\draw[-,thick,color=\figonecolorone!60] (10,0) -- (-3,-5);
	\draw[-,thick,color=\figonecolorone!40] (10,0) -- (-4,-2);
	\draw[-,thick,color=\figonecolorone!20] (10,0) -- (-5,-3);
	
	\draw[-,thick,color=\figonecolortwo!40] (0,10) -- (-1,-4);
	\draw[-,thick,color=\figonecolortwo] (0,10) -- (-2,-1);
	\draw[-,thick,color=\figonecolortwo!20] (0,10) -- (-3,-5);
	\draw[-,thick,color=\figonecolortwo!80] (0,10) -- (-4,-2);
	\draw[-,thick,color=\figonecolortwo!60] (0,10) -- (-5,-3);
	
	% Axes
	\draw[-stealth,thick] (-6, 0) -- (11, 0) node[right] {$x$};
	\draw[-stealth,thick] (0, -6) -- (0, 11) node[above] {$y$};
	
	% Points v_1 and v_2
	\filldraw[\figonecolorone] (10,0) circle (3pt) node[below right] {$\mbfv_1$};
	\filldraw[\figonecolortwo] (0,10) circle (3pt) node[above left] {$\mbfv_2$};
	
	% Points a_1 to a_5
	\filldraw[\figonecolorthree] (-1,-4) circle (3pt) node[above left] {$\mbfa_1$};
	\filldraw[\figonecolorthree] (-2,-1) circle (3pt) node[above right] {$\mbfa_2$};
	\filldraw[\figonecolorthree] (-3,-5) circle (3pt) node[below left] {$\mbfa_3$};
	\filldraw[\figonecolorthree] (-4,-2) circle (3pt) node[above right] {$\mbfa_4$};
	\filldraw[\figonecolorthree] (-5,-3) circle (3pt) node[below left] {$\mbfa_5$};
	
%	\filldraw[black] (0,0) circle (3pt) node[above right]{};
\end{tikzpicture}
	\caption{%
		Rank-preserving AR embedding of the preference profile presented in \cref{ex1}, using $c = 10$. 
		The line segments connecting the voters to the alternatives are colored from dark to light in decreasing order of their length, 
		highlighting that the preferences are in fact ordered by distances (which we can also verify arithmetically).
	}\label{fig:ar-embedding}
	\Description{TODO} % TODO accessibility description
\end{figure}
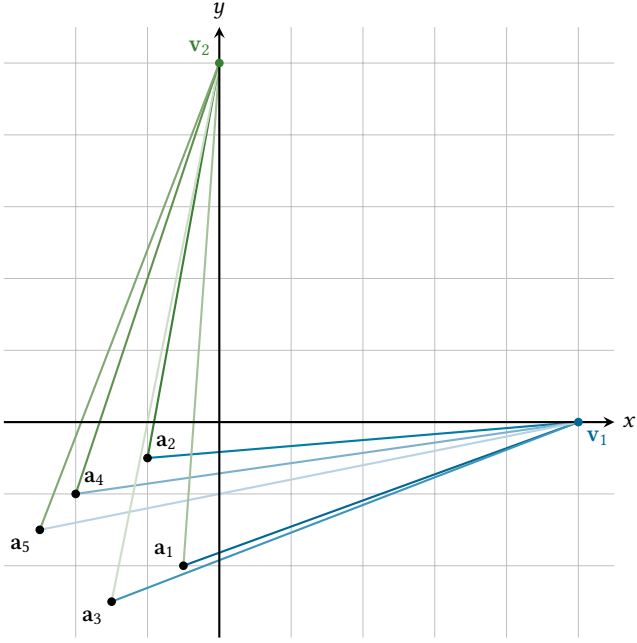

In the following, we first use AR embeddings to establish that for $p>1$, 
any preference profile with $n$ voters rank-embeds into $(\RR^n,\|\cdot\|_p)$. 
Afterwards, we discuss why this construction does not work for $p = 1$, 
sketching the alternative approach showcased by \citet{chen2022manhattan} in \Cref{def:max-rank-embedding}.

\begin{prop}[Dimensionality depending on $n$]\label{prop:prop1}
	\hspace*{0.1em}\\
	For $1< p\leq \infty$, any preference profile~~$\mcP$ with $n$ voters can be rank-embedded into $(\RR^n,\|\cdot\|_p)$, 
	where $\|\cdot\|_p$ denotes the $p$-norm.
\end{prop}
\begin{proof} 
	We show this by establishing that for any $p$ and a given preference profile~$\mcP$, 
	we can always choose $c>0$ such that the AR embedding $\mcE_c$ of~$\mcP$ is rank-preserving under the $p$-norm.
	We distinguish two cases, $1 < p < \infty$ and $p = \infty$.
	
	\emph{Case~1: $1 < p < \infty$.}\\
	For $1 < p< \infty$, 
	let $c>0$, 
	and take the corresponding AR embedding $\mcE_c$ of $\mcP$. 
	Now, $\mcE_c$ is rank-preserving for $\mcP$ if for all alternatives $j$, $l$ and all voters $i$, we have

	\begin{align*}
		a_j\succ_i a_l \Longrightarrow \|\mbfv_i-\mbfa_j\|_p< \|\mbfv_i-\mbfa_l\|_p\;,
	\end{align*}
	or equivalently, 
	\begin{align*}
		a_j\succ_i a_l \Longrightarrow \|\mbfv_i-\mbfa_j\|_p^p< \|\mbfv_i-\mbfa_l\|_p^p\;.
	\end{align*}
	For any $a_j\succ_i a_l$, we have 
	\begin{align*}
		\|\mbfv_i-\mbfa_j\|_p^p 
		&= \sum_{k=1}^{n}|a_j^{(k)}-v_i^{(k)}|^p
		=\sum_{k\neq i}|a_j^{(k)}|^p+|-\rk_{i}j-c|^p\\
		&=|c+\rk_ij|^p+\|\mbfa_j\|_p^p-|\rk_ij|^p\;,
	\end{align*}
	as well as
	\begin{align*}
		\|\mbfv_i-\mbfa_l\|_p^p=|c+\rk_il|^p+\|\mbfa_l\|_p^p-|\rk_il|^p\;.
	\end{align*} 
	Thus, $\mcE_c$ is rank-preserving for $\mcP$ if and only if for all alternatives~$j,l$ and all voters~$i$, $\mbfa_j\succ_i\mbfa_l$ is equivalent to 
	\begin{align*}
		 |c+\rk_ij|^p+\|\mbfa_j\|_p^p-|\rk_ij|^p<|c+\rk_il|^p+\|\mbfa_l\|_p^p-|\rk_il|^p\;.
	\end{align*} 
	
	Analyzing this inequality independently of $c$, 
	for fixed $i$, 
	the right-hand side is \emph{minimized}, 
	and the left-hand side is \emph{maximized},~for
	\begin{align*}
		a_l^{(k)} = \begin{cases}
			-\rk_il&k = i\\
			-1&k\neq i\;,\\
		\end{cases}\quad\text{and}\quad
		a_j^{(k)}=&\begin{cases}
			-\rk_ij&k = i\\
			-m&k \neq i\;.
		\end{cases}
	\end{align*}
	In this case, 
	\begin{align*}
		\|\mbfa_l\|_p^p=&(n-1)+(\rk_il)^p\;,~\text{and}\\
		\|\mbfa_j\|_p^p=&(n-1)m^p+(\rk_ij)^p\;.
	\end{align*}
	
	Hence, all we need is for $c$ to satisfy
	\begin{align*}
		&(c+\mathrm{rk}_ij)^p+(n-1)m^p+(\mathrm{rk}_ij)^p-(\mathrm{rk}_ij)^p\\
		<&(c+\mathrm{rk}_il)^p+(n-1)+(\mathrm{rk}_il)^p-(\mathrm{rk}_il)^p\;,
	\end{align*}
	 which reduces to
	 \begin{align*}
	 	(n-1)m^p-(n-1)<(c+\mathrm{rk}_il)^p-(c+\mathrm{rk}_ij)^p\;.
	 \end{align*}
	  The right-hand side is minimized when $\mathrm{rk}_il=2$ and $\mathrm{rk}_ij=1$. 
	  Thus, all we need is a $c$ such that 
	  \begin{align}
	  	(c+2)^p-(c+1)^p>(n-1)m^p-(n-1)\;.\label{eq:choosing-c}
	  \end{align}
	  Because the function $(x+2)^p-(x+1)^p$ is strictly increasing for $p>1$, such a $c$ must exist.
	%
%	Thus, because we consider the worst-case scenario, 
%	if $c$ satisfies
%	\begin{align*}
%		&(c+\rk_ij)^p+(n-1)n^p+(\rk_ij)^p-(\rk_ij)^p\\ <~& (c+\rk_il)^p+(n-1)+(\rk_il)^p-(\rk_il)^p\;,
%	\end{align*}
%	 then the AR embedding $\mcE_c$ is rank-preserving for any preference profile $\mcP$ with $n$ voters.
%	After simplifying and rearranging, this equation becomes 
%	\begin{align*}
%		(n-1)(n^p-1) < (c+\rk_il)^p-(c+\rk_ij)^p\;.
%	\end{align*}
%	Because $\rk_il>\rk_ij,$ as $a_j\succ_i a_l$, 
%	the right-hand side is minimized for $\rk_ij=1$ and $\rk_il=2$. 
%	Hence, it suffices to show that there exists $c>0$ such that
%	\begin{align}\label{eq:choosing-c}
%		(c+2)^p-(c+1)^p>(n-1)(n^p-1)\;.
%	\end{align}
	%
%	Now define $f(x)=(x+2)^p-(x+1)^p$. 
%	Because $p>1$, 
%	$f(x)$ is strictly increasing as $x\to\infty$, 
%	and since $n$ and $p$ are fixed, there must exist some $c$ such that $f(c)>(n-1)(n^p-1)$. 
    %
    Hence, for $1 < p < \infty$, a rank-preserving AR embedding for $\mcP$ must exist.
	
	\emph{Case 2: $p = \infty$.}\\ 
	For $p=\infty$, set $c=m$ and choose the corresponding AR embedding $\mcE_m$. 
	Consider any $a_j\succ_ia_l$. 
	Then 
    \begin{align*}
        \|\mbfv_i-\mbfa_j\|_\infty=\max\{c+\rk_ij,\rk_kj \mid 1\leq k\neq i \leq n\}\;.
    \end{align*}
	As $c=m$ and $\rk_kj>0$, 
	we have $c+\rk_ij>m\geq \rk_kj$ for all $k\neq i$, 
	$\|\mbfv_i-\mbfa_j\|_\infty=c+\rk_ij$, and 
    \begin{align*}
        \|\mbfv_i-\mbfa_j\|_\infty=c+\rk_ij<c+\rk_il=\|\mbfv_i-\mbfa_l\|_\infty\\
        \iff \rk_ij<\rk_il \iff a_j\succ_i a_l\;.
    \end{align*}
	Joining both cases, 
	we obtain that for all $1< p\leq \infty$, 
	every preference profile $\mcP$ with $n$ voters rank-embeds into $(\RR^n,\|\cdot\|_p)$ using an AR embedding.
\end{proof}

We note three interesting implications of the condition on $c$ expressed in \Cref{eq:choosing-c}. 
First, if $(c'+2)^p-(c'+1)^p=(n-1)(n^p-1)$, 
then for all $c>c'$, an AR embedding $\mcE_{c}$ will be rank-preserving for \emph{any} preference profile $\mcP$ with $n$ voters (although the constant necessary for any specific profile may be smaller). 
Second, if $p < q$, the worst-case $c$ required for a rank-preserving AR embedding under the $p$-norm will also work under the $q$-norm. 
Lastly, as $p$ approaches $1$ from above, 
$f(x) = (x+2)^p-(x+1)^p$ grows to infinity at a slower rate, 
so the $c$ needed to guarantee \Cref{eq:choosing-c} gets larger. 
More precisely, the $c$ required as $p\searrow 1$ grows exponentially in $\frac{1}{p-1}$. 
This is made precise in \Cref{lemma1}.  
%whose technical proof we defer to \Cref{apx:deferred-proofs}.

\begin{restatable}[Asymptotics of $c$ as a function of $\frac{1}{p-1}$]{prop}{casymptotics}\label{lemma1}
	\hspace*{0.1em}\\
	Let $n\in \NN$ and $p>1$, 
	and define 
	\begin{align*}
		c(x)\coloneq\inf\{c\mid (c+2)^{1+\frac{1}{x}}-(c+1)^{1+\frac{1}{x}}>(n-1)m^{1+\frac{1}{x}}-(n-1)\}\;.
	\end{align*}
	Then $c(\frac{1}{p-1})\in \Theta(\exp(\frac{1}{p-1}))$.
\end{restatable}
\begin{proof}
	Let $t = \frac{1}{p-1} \Leftrightarrow p=1+\frac{1}{t}$ for $t>0$ (and $p > 1$), 
	and set 
	\begin{align*}
		c(t):=\inf\{c\mid (c+2)^p-(c+1)^p>(n-1)(m^p-1)\}\;,
	\end{align*}
	for $n,m\in\NN$ fixed.
	Define the function $g(x)=(x+1)^p$ for $x > 0$. 
	
	By the mean value theorem, for any $x$, there exists some $\varepsilon\in [0,1]$ such that \[g'(x+\varepsilon)=p(x+1+\varepsilon)^{p-1}=(x+2)^p-(x+1)^p.\] 
	This $\varepsilon$ varies with each $x$ chosen, so we can consider $\varepsilon$ as a function of $x$. 
	We claim that $\varepsilon(x)$ is continuous. 
	To see this, we observe that 
	\begin{align*}
		p(x+1+\varepsilon(x))^{p-1}&=(x+2)^{p}-(x+1)^p\\
		\Leftrightarrow~p^{\frac{1}{p-1}}(x+1+\varepsilon(x))&=((x+2)^p-(x+1)^p)^{\frac{1}{p-1}}\\
		\Leftrightarrow~ x+1+\varepsilon(x)&=\frac{((x+2)^p-(x+1)^p)^{\frac{1}{p-1}}}{p^{\frac{1}{p-1}}}\\\Leftrightarrow~\varepsilon(x)&=\frac{((x+2)^p-(x+1)^p)^{\frac{1}{p-1}}}{p^{\frac{1}{p-1}}}-(x+1)\;.
	\end{align*}
	Because $p>1$ and $x>0$ by assumption, 
	$\varepsilon(x)$ is a continuous function. 
	
	Now define $d(t)=\varepsilon(c(t))$.
	Then, by definition, $g'(c(t)+d(t))$ can be written as the three expressions
	\begin{align*}
		g'(c(t)+d(t)) &=(1+\frac{1}{t})(c(t)+1+d(t))^\frac{1}{t}\\
		&=(c(t)+2)^{\frac{1}{t}}-(c(t)+1)^{\frac{1}{t}}\\
		&=(n-1)m^{1+\frac{1}{t}}-(n-1)\;.
	\end{align*} 
	Therefore, we get that
	\begin{align*}
		\lim_{t\to\infty}(1+\frac{1}{t})(c(t)+1+d(t))^\frac{1}{t}=(n-1)m-(n-1)=(m-1)(n-1)\;.
	\end{align*}
	Taking the logarithm of both sides, which is justified by positive continuity, 
	we get
	\begin{align*}
		\lim_{t\to\infty}\log(1+\frac{1}{t})+\frac{1}{t}\log(c(t)+d(t)+1)=\lim_{t\to\infty}\frac{1}{t}\log(c(t))=C\;,
	\end{align*} 
	where $C=\log((m-1)(n-1))$ as $d(t)\in [0,1]$.
%	
%	Evaluating $g'(x)$ for $x = c(t)+d(t)$, we get 
%	\begin{align*}
%		g'(c(t)+d(t)) &
%		= (1+\frac{1}{t})(c(t)+1+d(t))^{\frac{1}{t}}\\
%		&=(c(t)+2)^{1+\frac{1}{t}}-(c(t)+1)^{1+\frac{1}{t}}\\
%		&=(n-1)(n^{1+\frac{1}{t}}-1)\;.
%	\end{align*}
%	Our goal is to understand the behavior of $c(t)$ as $t\to\infty$. 
%	Taking the limit of the right-hand side as $t\to\infty$ becomes $(n-1)^2$, 
%	and because every function we are considering is continuous and monotone, the limit exists. 
%	Thus, we obtain that
%	\begin{align*}
%		\lim_{t\to\infty}(1+\frac{1}{t})(c(t)+1+d(t))^{\frac{1}{t}}=(n-1)^2\;.
%	\end{align*}
%	Taking the logarithm of both sides (which is justifiable as both sides are positive and continuous), 
%	\begin{align*}
%		\lim_{t\to\infty}\left(\log(1+\frac{1}{t})+\frac{1}{t}\log(c(t)+1+d(t))\right)=2\log(n-1)\;.
%	\end{align*}
%	Splitting up the limit into the sum, 
%	we notice that $\log(1+\frac{1}{t})\to 0$ as $t\to\infty$, 
%	so the result becomes 
%	\begin{align*}
%		\lim_{t\to\infty}\frac{1}{t}\log(c(t)+1+d(t))=2\log(n-1)\;.
%	\end{align*}
%	
%	Now, we know that as $t\to\infty$, $c(t)\to\infty$ and $d(t)\in [0,1]$ (as $\varepsilon \in [0,1]$), 
%	so we can rewrite this limit as 
%	\begin{align*}
%		\lim_{t\to\infty}\frac{1}{t}\log(c(t))=2\log(n-1)\;.
%	\end{align*}
%	Thus, we get $\log(c(t))\in\Theta(t)$, so $c(t)\in \Theta(\exp(t))$. 
%	As $t = \frac{1}{p-1}$, 
%	we have $c(\frac{1}{p-1})\in \Theta(\exp(\frac{1}{p-1}))$.
\end{proof}

Finally, \Cref{eq:choosing-c} also clarifies why the approach using AR embeddings does not work for the $1$-norm: 
For $p = 1$, the left-hand side becomes a constant~$1$. 
An alternative construction specifically for $p = 1$ was introduced by \citet{chen2022manhattan}. 
We refer to their embeddings as \emph{max-rank embeddings}---due to their dependency on maximum ranks. 

\begin{defn}[Max-Rank Embedding]\label{def:max-rank-embedding}
	Given a preference profile $\mcP$ with $m$ alternatives and $n$ voters as well as a constant $c>0$, 
	a max-rank embedding sets $\mbfv_i=m \mbfe_i$ and 
	\begin{align*}
		a_j^{(i)}=\begin{cases}\rk_ij-\mathrm{mk}_j & i=g_j\\c+2\rk_ij + \sum_{k=1}^n (\rk_kj-\mathrm{mk}_j) & i \neq g_j\;,\end{cases}
	\end{align*}
	where $g_j=\arg\max_i \rk_ij$ and $\mathrm{mk}_j=\max_i \rk_ij$.
\end{defn}
This definition ensures that when calculating $\|\mbfv_i-\mbfa_j\|_1$, 
the expression reduces to a linear function of $\rk_ij$,
showing that larger distances correspond to lower (i.e., larger-in-number) ranks. 
However, unlike AR embeddings, max-rank embeddings do not seem to generalize to all $p > 1$. 

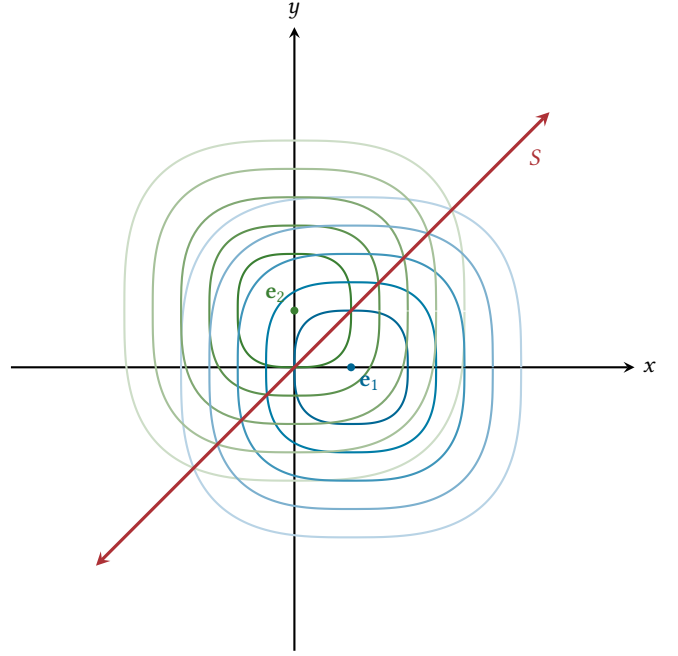
\begin{figure}[t]
	\centering
	% !TEX root = ../main.tex 
\begin{tikzpicture}[scale=0.75]
	\newcommand{\figtwocolorone}{MidnightBlue}
	\newcommand{\figtwocolortwo}{OliveGreen}
	\newcommand{\figtwocolorthree}{Maroon}
	\def\n{100} % Number of points to draw the curve
	\def\a{1}    % Parameter for the curve's shape
	
	% Optional: Grid
%	\draw[step=2cm,gray!50,very thin] (-4.5,-4.5) grid (5.5,5.5);
	
	\draw[-stealth,thick] (-5, 0) -- (6, 0) node[right] {$x$};
	\draw[-stealth,thick] (0, -5) -- (0, 6) node[above] {$y$};

	\draw[\figtwocolorone, thick] plot[domain=0:2*pi, samples=\n, smooth] ({1+abs(\a*cos(\x r))^(2/3)*sign(cos(\x r))}, {abs(\a*sin(\x r))^(2/3)*sign(sin(\x r))});
	\draw[\figtwocolortwo, thick] plot[domain=0:2*pi, samples=\n,smooth] ({abs(\a*cos(\x r))^(2/3)*sign(cos(\x r))}, {1+abs(\a*sin(\x r))^(2/3)*sign(sin(\x r))});
	
	\draw[\figtwocolorone!20, thick] plot[domain=0:2*pi, samples=\n,smooth] ({1+3*abs(\a*cos(\x r))^(2/3)*sign(cos(\x r))}, {3*abs(\a*sin(\x r))^(2/3)*sign(sin(\x r))});
	\draw[\figtwocolortwo!20, thick] plot[domain=0:2*pi, samples=\n,smooth] ({3*abs(\a*cos(\x r))^(2/3)*sign(cos(\x r))}, {1+3*abs(\a*sin(\x r))^(2/3)*sign(sin(\x r))});
	
	\draw[\figtwocolorone!80, thick] plot[domain=0:2*pi, samples=\n,smooth] ({1+1.5*abs(\a*cos(\x r))^(2/3)*sign(cos(\x r))}, {1.5*abs(\a*sin(\x r))^(2/3)*sign(sin(\x r))});
	\draw[\figtwocolortwo!80, thick] plot[domain=0:2*pi, samples=\n,smooth] ({1.5*abs(\a*cos(\x r))^(2/3)*sign(cos(\x r))}, {1+1.5*abs(\a*sin(\x r))^(2/3)*sign(sin(\x r))});
	
	\draw[\figtwocolorone!60, thick] plot[domain=0:2*pi, samples=\n,smooth] ({1+2*abs(\a*cos(\x r))^(2/3)*sign(cos(\x r))}, {2*abs(\a*sin(\x r))^(2/3)*sign(sin(\x r))});
	\draw[\figtwocolortwo!60, thick] plot[domain=0:2*pi, samples=\n] ({2*abs(\a*cos(\x r))^(2/3)*sign(cos(\x r))}, {1+2*abs(\a*sin(\x r))^(2/3)*sign(sin(\x r))});
	
	\draw[\figtwocolorone!40, thick] plot[domain=0:2*pi, samples=\n,smooth] ({1+2.5*abs(\a*cos(\x r))^(2/3)*sign(cos(\x r))}, {2.5*abs(\a*sin(\x r))^(2/3)*sign(sin(\x r))});
	\draw[\figtwocolortwo!40, thick] plot[domain=0:2*pi, samples=\n,smooth] ({2.5*abs(\a*cos(\x r))^(2/3)*sign(cos(\x r))}, {1+2.5*abs(\a*sin(\x r))^(2/3)*sign(sin(\x r))});
	
	\draw[stealth-stealth,color=\figtwocolorthree,very thick] (-3.5, -3.5) -- (4.5, 4.5);
	\fill[\figtwocolorthree] (4,4) circle (0pt) node[below right] {$S$};%{$(1,0)$};
	
	\fill[\figtwocolorone] (1,0) circle (2pt) node[below right] {$\mbfe_1$};%{$(1,0)$};
	\fill[\figtwocolortwo] (0,1) circle (2pt) node[above left] {$\mbfe_2$};%{$(0,1)$};
	
\end{tikzpicture}		
	\caption{%
		In $\RR^2$, the pairs of balls with equal radii under the $3$-norm around $\mbfe_1 = (1,0)$ (blue) and $\mbfe_2 = (0,1)$ (green) 
		intersect on a line (red).
	}\label{fig:hyperplane}
	\Description{TODO} % TODO accessibility description
\end{figure}

\subsection{Dimensionality Depending on the Number of Alternatives}

To establish rank embeddability as a function of $m$, 
we need to show that any preference profile $\mcP$ with $m$ alternatives can be embedded into $(\RR^{m-1},\|\cdot\|_p)$ for $1< p\leq \infty$. 
Again, the case $p=1$ is handled by \cite{chen2022manhattan}. 
While the case $p = \infty$ is treated separately in our proof, 
the generalization for $1 < p < \infty$ relies on the following technical lemma, 
which is implicitly used by \cite{bogomolnaia2007euclidean}. 
We provide some geometric intuition for this claim using $(\RR^2,\|\cdot\|_3)$ in \Cref{fig:hyperplane}.
% and defer the proof to \Cref{apx:deferred-proofs}.

\begin{restatable}[Hyperplane separating unit vectors]{lemma}{hyperplanelemma}\label{lemma2}
	\hspace*{0.1em}\\
	Let $1<p<\infty$ and $m>0$, 
	and let $\mbfe_i$ denote the $i^{\text{th}}$ unit vector of $\RR^m$. 
	Then for all $1\leq i\neq j \leq m$, the set 
    \begin{align*}
        S\coloneq\{\mbfx\in \RR^m\mid\|\mbfx-\mbfe_i\|_p=\|\mbfx-\mbfe_j\|_p\}
    \end{align*}
    is precisely the hyperplane $H$ parameterized by $\mbfx\in \RR^m$ such that $x_i=x_j.$
\end{restatable}
\begin{proof}
	We begin by establishing that $H\subseteq S$. 
	To this end, suppose that $\mbfx\in \RR^m$ with $x_i=x_j$. 
	Then, 
	\begin{align*}
		\|\mbfx-\mbfe_i\|_p^p=|x_i-1|^p+|x_j|^p+\sum_{\substack{1\leq k\leq m\\k\neq i,j}}|x_k|^p\;,
	\end{align*}
	and likewise, 
	\begin{align*}
		\|\mbfx-\mbfe_j\|_p^p=|x_j-1|^p+|x_i|^p+\sum_{\substack{1\leq k\leq m\\k\neq i,j}}|x_k|^p\;.
	\end{align*}
	Because $x_i=x_j$, 
	we have
	\begin{align*} 			|x_i-1|^p+|x_j|^p=|x_i|^p+|x_j-1|^p\;.
	\end{align*}
	Furthermore, because the sums are identical,
	\begin{align*} 		
		\|\mbfx-\mbfe_i\|_p^p=\|\mbfx-\mbfe_j\|_p^p\;, 
	\end{align*}
	and because the summands are absolute values, everything is positive, so 
	\begin{align*} 		
		\|\mbfx-\mbfe_i\|_p=\|\mbfx-\mbfe_j\|_p\;.
	\end{align*}
	
	To establish that $S \subseteq H$, fix some $\mbfx\in S$. 
	In this case, 
	$\|\mbfx-\mbfe_i\|_p=$ $\|\mbfx-\mbfe_j\|_p$, 
	so 
	\begin{align}\label{eq:sumsum}
		\|\mbfx-\mbfe_i\|_p^p=\|\mbfx-\mbfe_j\|_p^p\;,
	\end{align}
	where 
	\begin{align}\label{eq:sum1}
		\|\mbfx-\mbfe_i\|_p^p=|x_i-1|^p+|x_j|^p+\sum_{\substack{1\leq k\leq m\\k\neq i,j}}|x_k|^p\;,
	\end{align}
	and 
	\begin{align}\label{eq:sum2}
		\|\mbfx-\mbfe_j\|_p^p=|x_j-1|^p+|x_i|^p+\sum_{\substack{1\leq k\leq m\\k\neq i,j}}|x_k|^p\;.
	\end{align}
	As the summands in \Cref{eq:sum1,eq:sum2} are identical, 
	\Cref{eq:sumsum} reduces to 
	\begin{align*}
		&|x_i-1|^p+|x_j|^p=|x_j-1|^p+|x_i|^p\\ \Longleftrightarrow~& |x_j|^p-|x_j-1|^p=|x_i|^p-|x_i-1|^p\;.
	\end{align*}
	Now define $f(x)=|x|^p-|x-1|^p$ for $x>0$. 
	For $x\geq1$, this is clearly monotonic. 
	Moreover, for $0<x<y<1$, we have that $0<|y-1|<|x-1|<1$. 
	Then, 
	\begin{align*}
		|x|^p-|x-1|^p<|y|^p-|x-1|^p<|y|^p-|y-1|^p\;,
	\end{align*}
	so $f$ is monotonic for all $x > 0$. 
%	
%	As seen in the proof of \Cref{prop:prop1}, 
%	since we assume that $p>1$, 
%	$f(x)$ is strictly increasing. 
	Hence, for $s,t>0$, $f(s)=f(t)$ implies $s=t$. 
	Thus, we have $x_i=x_j$, and consequently, $S=H$.
\end{proof}

We note that this lemma relies on the chosen centers of $\mbfe_i$ and $\mbfe_j$ and does not generalize to two arbitrary points for $p\neq 2$.
Using \Cref{lemma2}, we can establish the desired rank-embeddability result depending on the number of alternatives. 

\begin{figure}[t]
	\centering
	% !TEX root = ../main.tex 
\tdplotsetmaincoords{70}{110}
\begin{tikzpicture}[scale=4.75,tdplot_main_coords]
	% Define the points
	\coordinate (A) at (1,0,0);
	\coordinate (B) at (0,1,0);
	\coordinate (C) at (0,0,1);
        \coordinate (v1) at (0.9,0.65,1.0);
	
	% Midpoints of the triangle's edges
	\coordinate (ABmid) at ($ (A)!0.5!(B) $);
	\coordinate (ACmid) at ($ (A)!0.5!(C) $);
	\coordinate (BCmid) at ($ (B)!0.5!(C) $);

    \coordinate (ABCmid) at ($ (BCmid)!0.33333!(A) $);
	
	% Triangle plane fill
	\filldraw[fill=MidnightBlue!10,opacity=0.5] (A) -- (B) -- (C) -- cycle;

    \fill[fill=Maroon!10,opacity=0.5] (C) -- (BCmid) -- (ABCmid) -- cycle;
	
	% Draw axes
	\draw[-stealth] (0,0,0) -- (1.2,0,0) node[anchor=north east]{$x$};
	\draw[-stealth] (0,0,0) -- (0,1.2,0) node[anchor=north west]{$y$};
	\draw[-stealth] (0,0,0) -- (0,0,1.2) node[anchor=south]{$z$};
	
	% Draw triangle edges
	\draw[thick] (A) -- (B) -- (C) -- (A);
	
	% Draw medians
	\draw[dashed,thick,Maroon] (A) -- (BCmid);
	\draw[dashed,thick,Maroon] (B) -- (ACmid);
	\draw[dashed,thick,Maroon] (C) -- (ABmid);
	
	% Draw points
	\filldraw[black] (A) circle (0.4pt) node[anchor=north west] {$\mbfa_1 = (1,0,0)$};
	\filldraw[black] (B) circle (0.4pt) node[anchor=south west] {$\mbfa_2 = (0,1,0)$};
	\filldraw[black] (C) circle (0.4pt) node[anchor=south west] {$\mbfa_3 = (0,0,1)$};
	
	\filldraw[black] (ABmid) circle (0.3pt);
	\filldraw[black] (ACmid) circle (0.3pt);
	\filldraw[black] (BCmid) circle (0.3pt);
    \filldraw[Maroon] (v1) circle (0pt) node[anchor=south west] {$a_3 \succ a_2 \succ a_1$};
	
\end{tikzpicture}
	\caption{%
		Median-based embedding of the alternatives for $m = 3$, 
		with the corresponding $2$-dimensional hyperplane shaded in blue 
		and the medians between any pair of alternatives drawn as dotted red lines. 
        With this setup, for example, voters with preference $a_3\succ a_2 \succ a_1$ can be placed in the red-shaded region of the hyperplane.
	}\label{fig:median-embedding}
	\Description{TODO} % TODO accessibility description
\end{figure}
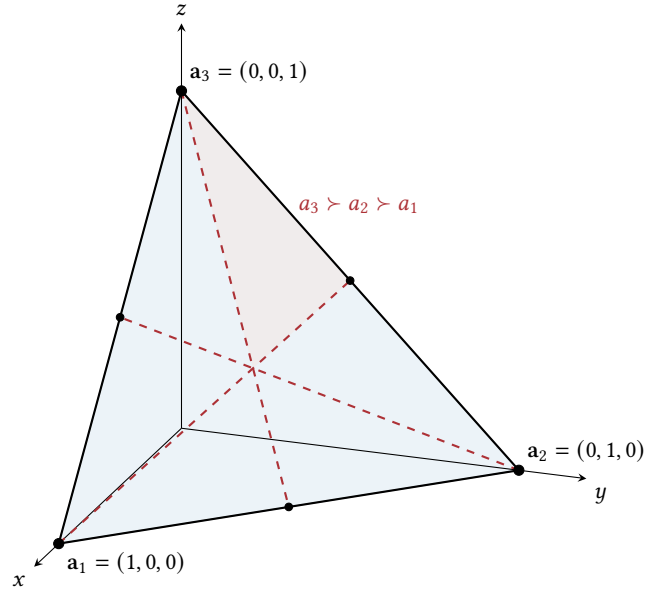

\begin{prop}[Dimensionality depending on $m$]\label{prop:prop2}
	\hspace*{0.1em}\\
	For $1 < p \leq \infty$, any preference profile $\mcP$ with $m$ alternatives can be rank-embedded into $(\RR^{m-1},\|\cdot\|_p)$.
\end{prop}
\begin{proof}
	We show that the following embedding construction introduced by \cite{bogomolnaia2007euclidean}, 
	which we call the \emph{median-based (MB) embedding}, 
	yields a rank-preserving embedding for any finite $p > 1$. 
	To create a median-based embedding, 
    we place the alternatives $\mbfa_1,\ldots,\mbfa_m$ at the basis vectors $\mbfe_1,\ldots,\mbfe_m$, 
	as illustrated for two dimensions in \Cref{fig:median-embedding}. 
	
	Again, we distinguish two cases, $1 < p < \infty$ and $p = \infty$. 
	
	\emph{Case 1: $1 < p < \infty$.}\\
	For $1 < p < \infty$, the points representing the $m$ alternatives in an MB embedding all lie on the $(m-1)$-dimensional affine hyperplane $P$ defined by $x_1+\cdots+x_m=1$. 
	By \Cref{lemma2},
	 the medians between any pair of alternatives are linear. 
	Hence, the preference ordering of a single voter can be represented as the intersection of the half spaces in which the preference lies. 
	This intersection will be non-empty because all medians intersect at the centroid, 
	and then any perturbation orthogonal to each median and in the direction of the preference will constitute a valid coordinate for the voter with such a profile. 
	
	\emph{Case 2: $p = \infty$.}\\
	For $p = \infty$, 
	\begin{align*}
		S\coloneq\{\mbfx\in \RR^m \mid \|\mbfx-\mbfe_i\|_\infty=\|\mbfx-\mbfe_j\|_\infty\}
	\end{align*}
	is not linear, 
 	and \Cref{lemma2} does not generalize to $p=\infty$. 
 	However, it turns out that $S\cap P$ is a linear subspace in $P$, 
 	which suffices to ensure rank embeddability. 

	To see this, 
	we can first restrict to the quadrant where each $x_i\geq 0.$ 
	Now assume that $x_i=x_j$ and $\mbfx\in P$. 
	Then 
	\begin{align*}
		&\max\{|x_i-1|,|x_j|,|x_k|\mid k\neq i,j\}\\=&\max\{|x_j-1|,|x_i|,|x_k|\mid k\neq i,j\}\;,
	\end{align*}
	so $\|\mbfx-\mbfe_i\|_\infty=\|\mbfx-\mbfe_j\|_\infty$.
	 
	Now assume that $\mbfx\in S\cap P$. 
	As we consider only the positive face of the simplex, we have
	\begin{align}\label{eq:maxcoord}
		\max\{1-x_i,x_j,x_k\mid k\neq i,j\}=\max\{x_i,1-x_j,x_k\mid k\neq i,j\}\;.
	\end{align}
	 Finally, suppose that some $x_k$ (for $k\neq i,j$) is the largest value of the elements in the set on the left-hand side of \Cref{eq:maxcoord}.  
	 Then $1-x_i\leq x_k$. 
	 However, $x_k+x_i \leq 1$ by the condition of the hyperplane, so $1-x_i\leq x_k \leq 1-x_i$ and $x_k=1-x_i$. 
	 By the same argument, the value of $x_k$ must be the maximum value of the right-hand side of \Cref{eq:maxcoord}, 
	 so $1-x_j \leq x_k$ and $x_j+x_k\leq 1$, such that $x_k=1-x_j$. 
	 Hence, $1-x_i=1-x_j$, and consequently, $x_i=x_j$. 
	 If the maximum value is attained at $x_j$, 
	 then the same argument applies with $1-x_i\leq x_j$ but $x_i+x_j\leq 1$, 
	 so $1-x_i\leq x_j\leq 1-x_i$ so $x_i+x_j=1$, and by symmetry, $x_i=x_j$ must hold. 
	 Thus, the medians will all be the spaces $x_i=x_j$ inside $P$.

     Joining both cases, we obtain that for all $1 < p \leq \infty$, every preference profile $\mcP$ with $m$ alternatives rank-embeds into $(\RR^{m-1},\|\cdot\|_p)$ using a median-based embedding. 
\end{proof}

Our proof of \Cref{prop:prop2} shows that there are only a few key requirements for a norm to permit rank-preserving median-based embeddings, 
simultaneously exposing which details from Proposition~5 by \citet{bogomolnaia2007euclidean} are strictly necessary.
Notably, the construction used in \Cref{prop:prop2} again cannot work for $p=1$, 
as the medians given by $\|\mbfx-\mbfe_i\|_1=\|\mbfx-\mbfe_j\|_1$ are neither linear nor otherwise well-behaved. 
It is thus not surprising that \citet{chen2022manhattan} use a very different construction, 
similar in flavor to the max-rank embeddings specified in \cref{def:max-rank-embedding}. 
We call this construction \emph{rank-pivot embedding} (because alternative $m$ can be viewed as a pivot) and restate it here for completeness. 

\begin{defn}[Rank-Pivot Embedding]\label{def:rank-pivot-embedding}
    Given a preference profile $\mcP$ with $m$ alternatives and $n$ voters, 
    for $k \in \{1, \dots, m-1\}$, $j \in \{1,\dots, m\}$, and $i \in \{1,\dots, n\}$,  
    a Rank-Pivot Embedding sets 
    \begin{align*}
        a_j^{(k)} =& \begin{cases}
            2m&j = k~\text{and}~j < m\\
            0&j\neq k~\text{or}~j = m\;,\text{ and}
        \end{cases}\\
        v_i^{(k)} =& \begin{cases}
            2m-\rk_i k&\rk_i k < \rk_i m\\
            m-\rk_i k&\rk_i k > \rk_i m\;.
        \end{cases}
    \end{align*}
\end{defn}

Proving \Cref{thm1} is now straightforward. 
\lpnorms
\begin{proof}
	For $1< p\leq \infty$, 
	any preference profile $\mcP_{A,V}$ rank-embeds into $(\RR^n,\|\cdot\|_p)$ by \cref{prop:prop1}, 
	and any preference profile $\mcP$ with $m$ alternatives rank-embeds into $(\RR^{m-1},\|\cdot\|_p)$ by \cref{prop:prop2}. 
	The case $p = 1$ is treated separately by \cite{chen2022manhattan}.
	Observing that any preference profile that rank-embeds into $\RR^d$ also rank-embeds into $\RR^{d'}$ for any $d'>d$, 
	the result follows.
\end{proof}
% !TeX spellcheck = en_US
% !TEX root = ../main.tex 
\section{Norm-Independent Rank-Preserving Embeddings into  \texorpdfstring{$\RR^2$}{TEXT}}\label{sec:thm2}

The main goal of this investigation was initially to extend the result inspiring \Cref{thm1} to \emph{any} norm. 
While the general problem remains out of reach with our current techniques, 
we \emph{can} obtain the desired result for any preference profile with \emph{two (types of) voters}, 
which should help generalize studies of facility-location problems like those by \citet{larson1983facility} or \citet{chan2021mechanism}. 
This still constitutes significant progress,  
and it should provide some foundation for future research into further abstraction. 

\begin{restatable*}[Rank embeddability for two {[}types of{]} voters under arbitrary norms]{thm}{twovoters}\label{thm2}
	Given $m$ alternatives, let $\mcP$ be a preference profile featuring two (types of) voters. 
	Then $\mcP$ rank-embeds into $(\RR^2,\|\cdot\|)$ for any norm $\|\cdot\|$ on~$\RR^2$.
\end{restatable*}

Our proof relies on the following technical lemma, 
the gist of which we visualize in \Cref{fig:annulus}. 
%The proof is deferred to \Cref{apx:deferred-proofs}.

\begin{restatable}[Placing two voters and three alternatives in {$\RR^2$} under arbitrary norms]{lemma}{annuluslemma}\label{lemma3}
	Let~$\mbfv_1,\mbfv_2$ be distinct points in $(\RR^2,\|\cdot\|)$, 
	where $\|\cdot\|$ denotes any norm on $\RR^2$. 
	Let $B(\mbfv_1,r_1)\subsetneq B(\mbfv_1,r_2)$ denote two balls under $\|\cdot\|$, 
	where 
	\begin{align*}
		0<r_1<r_2<\|\mbfv_1-\mbfv_2\|~\text{~such that~}~\mbfv_2\in \left(\overline{B(\mbfv_1,r_2)}\right)^c\;,
	\end{align*}
	with $\overline{S}$ denoting the closure and $S^c$ denoting the complement of a set $S$. 
	Given $p_1\in \partial B(\mbfv_1,r_1)$ and $\mbfp_2\in \partial B(\mbfv_2,r_2)$ such that $\|\mbfv_2-\mbfp_1\|\neq$ $\|\mbfv_2-p_2\|$, 
	where $\partial S$ denotes the boundary of a set~$S$, 
	let 
	\begin{align*}
		s_1\coloneq&\min\{\|\mbfv_2-\mbfp_1\|,\|\mbfv_2-\mbfp_2\|\}\text{~and~}\\s_2\coloneq&\max\{\|\mbfv_2-\mbfp_1\|,\|\mbfv_2-\mbfp_2\|\}\;.
	\end{align*}
	Now denote the annulus around $\mbfv_2$ between radius $s_1$ and $s_2$ as 
	\begin{align*}
		A(\mbfv_2,s_1,s_2)\coloneq\{\mbfx\in \RR^2 \mid s_1<\|\mbfx-\mbfp_2\|<s_2\}\;.
	\end{align*}
	Then
	 \begin{align*}
	 	S\coloneq\left(\overline{B(\mbfv_1,r_2)}\right)^c \cap A(\mbfv_2,s_1,s_2)\cap B(\mbfv_1,\|\mbfv_1-\mbfv_2\|)
	 \end{align*}
	 is not empty.\end{restatable}
\begin{proof} 
	Assume without loss of generality that $s_1=\|\mbfv_2-\mbfp_1\|$---%
	for $s_1=\|\mbfv_2-\mbfp_2\|$, 
	we simply change the use of the triangle inequality by swapping the points $\mbfp_1$ and $\mbfp_2$ in the argument that follows.
	Now choose $r_3$ such that $r_2<r_3<\|\mbfv_1-\mbfv_2\|$. 
	Observe that for $\mbfs\in \partial B(\mbfv_1,r_3)$, 
	$\|\mbfv_2-\mbfs\|$ is continuous, 
	and at 
	\begin{align*}
		\mbfs=\mbfv_1+\frac{r_3}{\|\mbfv_2-\mbfv_1\|}(\mbfv_2-\mbfv_1)\in \partial B(\mbfv_1,r_3)\;,
	\end{align*}
	we have 
	\begin{align*}
		\|\mbfv_2-\mbfs\|&=\|(\mbfv_2-\mbfv_1)-\frac{r_3}{\|\mbfv_2-\mbfv_1\|}(\mbfv_2-\mbfv_1)\|\\
		&=(1-\frac{r_3}{\|\mbfv_2-\mbfv_1\|})\|\mbfv_2-\mbfv_1\|\\
		&=\|\mbfv_2-\mbfv_1\|-r_3\;.
	\end{align*}
	Furthermore, at 
	\begin{align*}
		\mbfs=\mbfv_1-\frac{r_3}{\|\mbfv_2-\mbfv_1\|}(\mbfv_2-\mbfv_1)\in \partial B(\mbfv_1,r_3)\;,
	\end{align*}
	we have 
	\begin{align*}
		\|\mbfv_2-\mbfs\|&=\|\mbfv_2-(\mbfv_1-\frac{r_3}{\|\mbfv_2-\mbfv_1\|}(\mbfv_2-\mbfv_1))\|\\
		&=\|(\mbfv_2-\mbfv_1)+\frac{r_3}{\|\mbfv_2-\mbfv_1\|}(\mbfv_2-\mbfv_1)\|\\
		&=(1+\frac{r_3}{\|\mbfv_2-\mbfv_1\|})\|\mbfv_2-\mbfv_1\|=\|\mbfv_2-\mbfv_1\|+r_3\;.
	\end{align*}
	Now,
	\begin{align*}
		\|\mbfv_2-\mbfv_1\|&<\|\mbfv_2-\mbfp_1\|+\|\mbfp_1-\mbfv_1\|=r_1+\|\mbfv_2-\mbfp_1\|\\
		&<r_2+\|\mbfv_2-\mbfp_1\|<r_3+\|\mbfv_2-\mbfp_1\|\;,
	\end{align*}
	so 
	\begin{align*}
		\|\mbfv_2-\mbfv_1\|-r_3<s_1<s_2\;.
	\end{align*}
	Moreover,
	\begin{align*}
		s_2=~&\|\mbfv_2-p_2\|\\
		<~&\|\mbfv_2-\mbfv_1\|+\|\mbfv_1-\mbfp_2\|=\|\mbfv_2-\mbfv_1\|+r_2
		<~\|\mbfv_2-\mbfv_1\|+r_3\;,
	\end{align*} 
	so altogether, 
	\begin{align*}
		\|\mbfv_2-\mbfv_1\|-r_3<\mbfs_1<\mbfs_2<\|\mbfv_2-\mbfv_1\|+r_3\;.
	\end{align*}
	Thus, by continuity of $\|\mbfv_2-\mbfs\|$ in $\mbfs$, 
	there must exist $\mbfs$ such that $s_1<\|\mbfv_2-\mbfs\|<s_2$, 
	where $\mbfs\in \partial B(\mbfv_1,r_3).$ 
	Hence, there exists $\mbfs\in S$ as constructed. 
\end{proof}

\begin{figure}[t]
	\centering
	% !TEX root = ../main.tex 
\begin{tikzpicture}[scale=0.75]
	\newcommand{\ellipseonecolor}{MidnightBlue}
	\newcommand{\ellipsetwocolor}{OliveGreen}
	\newcommand{\annuluscolor}{Maroon}
	\def\n{100} % Number of points to draw the curve
	\draw[-stealth,thick] (-5, 0) -- (6, 0) node[right] {$x$};
	\draw[-stealth,thick] (0, -5) -- (0, 6) node[above] {$y$};
	\fill[\ellipseonecolor] (4,0) circle (2pt) node[above right] {$\mbfv_1$};
	\fill[\ellipsetwocolor] (0,4) circle (2pt) node[above left] {$\mbfv_2$};
	\draw[\ellipseonecolor, thick,shift={(4,0)},rotate=45] plot[domain=0:2*pi, samples=\n, smooth] ({2*cos(\x r))}, {sin(\x r)});
	\draw[\ellipsetwocolor, thick,shift={(0,4)},rotate=45] plot[domain=0:2*pi, samples=\n, smooth] ({2*cos(\x r))}, {sin(\x r)});

	\def\a{2}       % x radius
	\def\b{1}       % y radius
	\def\theta{45}  % rotation in degrees
	\def\cx{4}      % x center
	\def\cy{0}      % y center
	\def\t{195}      % parametric angle on ellipse (degrees)
	\pgfmathsetmacro{\xlocal}{\a * cos(\t)}
	\pgfmathsetmacro{\ylocal}{\b * sin(\t)}
	\pgfmathsetmacro{\xtemp}{\xlocal*cos(\theta) - \ylocal*sin(\theta)}
	\pgfmathsetmacro{\ytemp}{\xlocal*sin(\theta) + \ylocal*cos(\theta)}
	\pgfmathsetmacro{\xfinal}{\cx + \xtemp}
	\pgfmathsetmacro{\yfinal}{\cy + \ytemp}
	\pgfmathsetmacro{\xfinalone}{\xfinal}
	\pgfmathsetmacro{\yfinalone}{\yfinal}
	
	\draw[\ellipseonecolor!80, thick,shift={(4,0)},rotate=45] plot[domain=0:2*pi, samples=\n, smooth] ({2.5*cos(\x r))}, {1.25*sin(\x r)});
%	\draw[\ellipsecolor!60, thick,shift={(4,0)},rotate=45] plot[domain=0:2*pi, samples=\n, smooth] ({2.5*cos(\x r))}, {1.25*sin(\x r)});
	\def\s{37.5}
	\def\c{2.5}
	\def\d{1.25}
	\pgfmathsetmacro{\xlocal}{\c * cos(\s)}
	\pgfmathsetmacro{\ylocal}{\d * sin(\s)}
	\pgfmathsetmacro{\xtemp}{\xlocal*cos(\theta) - \ylocal*sin(\theta)}
	\pgfmathsetmacro{\ytemp}{\xlocal*sin(\theta) + \ylocal*cos(\theta)}
	\pgfmathsetmacro{\xfinal}{\cx + \xtemp}
	\pgfmathsetmacro{\yfinal}{\cy + \ytemp}
	
	%\draw[\ellipseonecolor!60, thick,shift={(4,0)},rotate=45] plot[domain=0:2*pi, samples=\n, smooth] ({3*cos(\x r))}, {1.5*sin(\x r)});
	\draw[\annuluscolor,thick,shift={(0,4)},rotate=45] plot[domain=5*pi/4:13*pi/8, samples=\n, smooth] ({10*cos(\x r))}, {5*sin(\x r)});
	\draw[\annuluscolor,thick,shift={(0,4)},rotate=45] plot[domain=10.25*pi/8:12.75*pi/8, samples=\n, smooth] ({12*cos(\x r))}, {6*sin(\x r)});
	
	\fill[black] (\xfinalone,\yfinalone) circle (2pt) node[below right] {$\mbfp_1$};
	\fill[black] (\xfinal,\yfinal) circle (2pt) node[below right] {$\mbfp_2$};
	\fill[BurntOrange] (\xfinalone-1,\yfinalone) circle (2pt) node[below left] {$\mbfp_3$};
\end{tikzpicture}
	\caption{Geometric intuition for \Cref{lemma3}. 
	Let the unit ball of our norm be given by a symmetric convex body like the green ellipse centered around $\mbfv_2$.
	Given points $\mbfp_1,\mbfp_2,\mbfp_3$ in $\RR^2$, 
	with $\mbfp_1,\mbfp_2$ already placed in the plane with $\mbfv_1$ and $\mbfv_2$, 
	we seek to place $\mbfp_3$ such that $\mbfv_1$ ranks $\mbfp_3$ last but $\mbfv_2$ ranks $\mbfp_3$ between $\mbfp_1$ and $\mbfp_2$. 
	Here, the red arcs represent the balls around $\mbfv_2$ that 
	touch $\mbfp_1$ and $\mbfp_2$, creating the relevant annulus around $\mbfv_2$. 
	To ensure that $\mbfp_3$ lies between $\mbfp_1$ and $\mbfp_2$ in the ranking of $\mbfv_2$, 
	we need to place the point inside this annulus.
	However, we also need to guarantee that $\mbfp_3$ is ranked last for $\mbfv_1$. 
	As demonstrated in the drawing, 
	we can achieve this by placing $\mbfp_3$ outside the blue balls, yet between the red arcs.  
	\Cref{lemma3} ensures that such a point always exists for any norm. 
	}\label{fig:annulus}
	\Description{TODO} % TODO accessibility description
\end{figure}
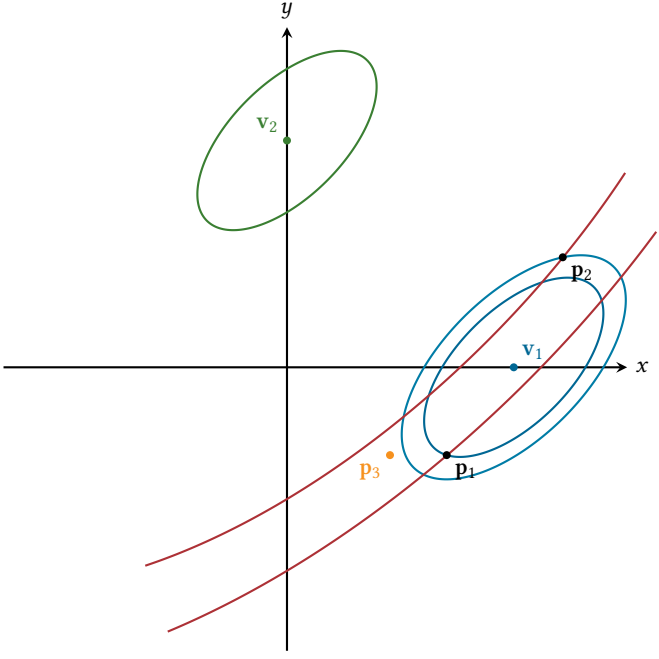

In our proof of \Cref{thm2}, 
we use induction to successively place each point in the ranking of $v_1$ in the area that preserves the required relationship to~$v_2$. 

\twovoters
\begin{proof}
	We proceed by induction on $m$. 
	Our induction will be on the hypothesis that for any preference profile $\mcP$ with $m>1$ alternates (the case $m=1$ is trivial) and two distinct voters $v_1$ and $v_2$ placed at $\mbfv_1, \mbfv_2\in \RR^2$, 
	there exist coordinates $\mbfa_1,\ldots,\mbfa_m$ such that the placement of $a_1,\ldots,a_m$ preserves the preference ordering~while 
	\begin{align}\label{eq:ball-placement}
		\mbfv_2\in \left(\overline{B(\mbfv_1,\max\{\|\mbfv_1-\mbfa_i\|\mid 1\leq i \leq m\})}\right)^c\;.
	\end{align}
	While the constraint expressed in \Cref{eq:ball-placement} is not essential, it makes the geometry much easier to work with, 
	which will facilitate the inductive step.
	
	\emph{Base case.}\\
	Suppose that $\mbfv_1$ and $\mbfv_2$ are any two distinct points in $\RR^2$. 
	If $m=2$, then there exist only two possible preference profiles up to symmetry: 
	(1)~$a_1 \succ_1 a_2$ and $a_1\succ_2 a_2$, or (2)~$a_1\succ_1 a_2$ and $a_2\succ_2 a_1$ (these are the only two classes of preference profiles up to relabeling). 
	In the first case, 
	we can place $\mbfa_1$ anywhere such that $\mbfv_2\in \left(\overline{B(\mbfv_1,\|\mbfv_1-\mbfa_1\|)}\right)^c$, 
	and then place $\mbfa_2$ anywhere further away from both $\mbfv_1$ and $\mbfv_2$ inside $\RR^2$ (specifically in $\left(\overline{B(\mbfv_1,\|\mbfv_1-\mbfa_1\|)}\right)^c \cap$ $\left(\overline{B(\mbfv_2,\|\mbfv_2-\mbfa_1\|)}\right)^c\neq\emptyset$). 
	In the second case, 
	we can place $\mbfa_1$ arbitrarily close to $\mbfv_1$ on the line from $\mbfv_1$ to $\mbfv_2$ and $\mbfa_2$ arbitrarily close to $\mbfv_1$. 
	In either case, we also have 
	\begin{align*}
		\mbfv_2\in \left(\overline{B(\mbfv_1,\max\{\|\mbfv_1-\mbfa_i\||1\leq i \leq 2\})}\right)^c\;.
	\end{align*}
	
	\emph{Inductive step.}\\ 
	Now assume that our induction hypothesis holds for all $k<m$. 
	Our goal is to show that it also holds for $k=m$. 
	We can safely assume that $v_1$ has preference ordering $a_1\succ_1 \cdots \succ_1 a_m$, 
	as this just amounts to a permutation on the labels of the candidates. 
	Now, the assumption that $\mbfv_2\in \left(\overline{B(\mbfv_1,\max\{\|\mbfv_1-\mbfa_i\|\mid 1\leq i \leq m\})}\right)^c$ 
	becomes simply that $\mbfv_2\in \left(\overline{B(\mbfv_1,\|\mbfv_1-\mbfa_k\|)}\right)^c.$
	
	By the inductive hypothesis, $\mbfa_1,\ldots,\mbfa_{m-1}$ are placed in $\RR^2$ such that 
	%\begin{align*}
	$	\|\mbfv_1-\mbfa_1\|<\cdots <\|\mbfv_1-\mbfa_{m-1}\|$, %\;,
	%\end{align*}
	and such that the ordering with respect to the preferences of $\mbfv_2$ is preserved. 
	
	To place $\mbfa_m$, we now distinguish two cases.
	 
	First, suppose that $a_m$ does \emph{not} rank last for $v_2$. 
	That is, there exist $i,j$ such that $a_i\succ_2a_m\succ_2a_j$. 
	Since $\mbfv_2\in \left(\overline{B(\mbfv_1,\mbfa_{m-1})}\right)^c$ and $\|\mbfv_1-\mbfa_{\min\{i,j\}}\|<\|\mbfv_1-\mbfa_{\max\{i,j\}}\|$, 
	\Cref{lemma3} can be applied with $\mbfp_1=\mbfa_{\min\{i,j\}}$, 
	$\mbfp_2=\mbfa_{\max\{i,j\}}$, 
	$r_1=\|\mbfv_1-\mbfa_{\min\{i,j\}}\|$, 
	and $r_2=\|\mbfv_1-\mbfa_{\max\{i,j\}}\|$ 
	(we use the $\min$/$\max$ notation here because the ordering of $i$ and $j$ depends on whether $i<j$ for $v_1$ but not for $v_2$). 
	By \Cref{lemma3}, there thus exists an $\mbfa_m$ such that
	\begin{align*}
		\|\mbfv_1-\mbfa_{\min\{i,j\}}\|<\|\mbfv_1-\mbfa_{\max\{i,j\}}\|<\|\mbfv_1-\mbfa_m\|\,
	\end{align*} 
	but also
	\begin{align*}
		\|\mbfv_2-\mbfa_i\|<\|\mbfv_2-\mbfa_m\|<\|\mbfv_2-\mbfa_j\|\;,
	\end{align*}
	where we still have that $\mbfv_2\in \left(\overline{B(\mbfv_1,\|\mbfv_1-\mbfa_m\|)}\right)^c$.
	
	Now suppose that $a_m$ ranks last for both $v_1$ \emph{and} $v_2$. 
	In this case, we can place $a_m$ at  $\mbfa_m=\mbfv_1+(\mbfv_1-\mbfv_2)(1-\varepsilon)$ 
	for some sufficiently small $\varepsilon$. 
	By construction, the distance from $\mbfa_m$ to $\mbfv_1$ and $\mbfv_2$ will be larger than that of any $\mbfa_i$, $i<m$, 
	while still preserving that $\mbfv_2\in \left(\overline{B(\mbfv_1,\|\mbfv_1-\mbfa_{m}\|)}\right)^c$.
\end{proof}

While this result is very strong for two (types of) voters in $\RR^2$, 
it will be hard to generalize to higher dimensions.  
The problem is that when $n$ increases, \Cref{lemma3} is no longer viable,  
because its construction strictly depends on having only one ordering ``free" and the other one ``fixed''. 
Once we have to contend with a third voter and are given multiple relationships of the form $a_{i}\succ_2a_m\succ_2a_j$ and $a_k\succ_3a_m\succ_3a_l$, 
it appears implausible to guarantee that the intersection of the two annuli will lie outside $B(\mbfv_1,\|\mbfv_1-\mbfa_{m-1}\|)$, making induction difficult. 
Therefore, we believe that working with general norms in higher dimensions will require novel techniques. 
% !TeX spellcheck = en_US
% !TEX root = ../main.tex 
\section{Discussion and Conclusion}\label{sec:discussion}

In this work, we made progress toward characterizing the general conditions under which ordinal preferences can be faithfully represented in real space, 
introducing the notion of \emph{rank embeddability} to capture this property without reference to an underlying norm. 

\textbf{Theoretical Results and Implications.}\quad
Our \emph{first result} (\Cref{thm1}) showed that any preference profile with $m$ alternatives and $n$ voters rank-embeds into $(\RR^d,\|\cdot\|_p)$ for $p \geq 1$ if $d \geq \min\{n,m-1\}$, 
generalizing the constructions introduced by \citet{bogomolnaia2007euclidean} from the Euclidean norm to all $p$-norms for $p>1$.
As a byproduct, we also 
expounded why these constructions do not generalize to the $p=1$ case,  
which \citet{chen2022manhattan} later addressed using a much more specialized approach, 
and streamlined the previously disparate treatment of the cases $p = 2$ and $p = \infty$ (see~\cite{bogomolnaia2007euclidean}~vs.~\cite{chen2022manhattan}). 

With our \emph{second result} (\Cref{thm2}), 
we established that any~preference profile with two (types of) voters rank-embeds into $(\RR^2,\|\cdot\|)$, for any norm $\|\cdot\|$. 
Our inductive proof relies on the intrinsic geometry of convex balls in $\RR^2$, 
and hence will be hard to generalize to higher dimensions. 
The result for $n = 2$, however, 
might still be useful in the context of facility-location problems---%
or in highly polarized settings, where two-type electorates might naturally arise. 
%While our construction is not explicit---%
%we use the continuity of the norm to guarantee that a point with the required properties will exist between the preference annulus of $v_2$ and outside all the preferences of $v_1$---, 
%in most technical situations, the norms given have very nice parameterizations, 
%and upon visualization (or merely by applying the parameterization arithmetically), 
%a point for the preference should be straightforward to choose.

The main implication of our results is that the behavior of spatial preferences under \emph{non-standard norms} merits systematic investigation.  
This may eventually build a case for application-sensitive norm selection based on (expected) features of the input data. 
For example, \citet{thorburn2023error} show that under the $2$-norm, 
almost all preference profiles fail to rank-embed into low dimensions ($d \ll \min\{n,m\}$). 
They raise concerns that using the $2$-norm, 
e.g., in recommender systems 
may yield inaccurate numerical preference representations.  
While it remains unclear \emph{if} and \emph{when} switching to a different norm could improve accuracy,  
our results underscore that these are questions worth asking. 

\textbf{Limitations and Future Work.}\quad
Our contribution focuses on the \emph{theoretical analysis} of spatial-preference models under arbitrary norms. 
Future work could systematically explore complementary approaches in this setting, such as \emph{simulations} and \emph{empirical studies}.

On the theoretical side, 
although our work made considerable progress toward a \emph{full characterization} of rank embeddability under general norms, 
we stopped short of solving the general case. 
%However, we hope that the advances we made in Theorems~\ref{thm1} and~\ref{thm2} 
%will catalyze progress toward a \emph{full characterization} of rank embeddability under general norms. 
Therefore, we explicitly pose the following conjecture. 
\begin{conjecture}[Rank-Embeddability Conjecture]\label{conj}
	\hspace*{0.1em}\\For $d\geq \min\{n,m-1\}$, 
	any preference profile $\mcP$ with $m$ alternatives and $n$ voters can be rank-embedded into $(\RR^d,\|\cdot\|)$, 
	where $\|\cdot\|$ denotes any norm.
\end{conjecture}

On the path toward full generalization, 
extending \Cref{thm1} to \emph{polynomial norms} (including weighted $p$-norms) could be a natural next step. 
Furthermore, to our knowledge, the techniques used by \citet{chen2022manhattan} to handle the case $p = 1$ do not generalize to $p$-norms for $p > 1$.
Thus, 
%while we managed to handle the $p = \infty$ case, treated separately in prior work, along with all other $p > 1$, 
defining a \emph{unified construction} that works for all $p$-norms also remains an open problem.

Beyond generalizing the upper bound on the dimensionality required to ensure rank embeddability, 
the \emph{tightness of the bound} merits further investigation. 
\citet{bogomolnaia2007euclidean} show that in the presence of \emph{indifferences} between alternatives 
(i.e., $a_j \sim_i a_l$ for some voter $i$ and alternatives $j\neq l$), 
for any $d$, there always exist preference profiles with $d+1$ voters or $d+2$ alternatives that cannot be embedded into $\RR^d$ under the Euclidean norm. 
However, the counterexample they give relies on a setting in which at least one voter is required to be equidistant to multiple alternatives in the embedding space. 
This criterion makes the tightness of the bound significantly easier to address, and it does not cover, e.g., preference profiles that contain only strict inequalities or incomplete preferences. 
Hence, it would be valuable to understand whether the bound is tight also in these scenarios,  
and if the other results by \citet{bogomolnaia2007euclidean} that rely on equalities can be generalized as well.

Scrutinizing \emph{low-dimen\-sional profiles} under arbitrary norms could constitute another interesting avenue for future work. 
\citet{escoffier2024euclidean} provide some good combinatorial results for the amount of preference profiles that are Euclidean, Manhattan, or $\ell^\infty$ in low dimensions, 
while \citet{chen2022manhattan} highlight just how distinct the resulting embeddings can be. 
Against this background, 
studying which preference profiles rank-embed under \emph{multiple norms} in a specific dimension might yield interesting insights. 
Given the generality of our result for $\RR^2$ (\cref{thm2}), 
some of our techniques might prove useful for tackling this question.

Several other approaches also hold promise for approaching the general case (\Cref{conj}). 
For instance, one could consider \emph{approximations of balls} via polynomial norms. 
There are really robust uniformity results for any convex symmetric ball via polynomials---%
see, e.g., the works by \citet{goodfellow2018machine} or \citet{totik2013approximation}---,
and we might be able to leverage those results for the higher-dimensional setting. 
While the corresponding constructions would not be explicit, 
any result establishing rank embeddability under arbitrary norms would provide a strong theoretical foundation for future work on problems related to spatial preferences.

\section*{Related Version}

An Extended Abstract of this work was published at AAMAS 2026~\cite{zeitlin2026real}.

\begin{acks}
	This research was supported by the Aalto Science Institute Summer Research Program.
\end{acks}

\balance
\bibliographystyle{ACM-Reference-Format} 
\bibliography{bibliography}

%%% -*-BibTeX-*-
%%% Do NOT edit. File created by BibTeX with style
%%% ACM-Reference-Format-Journals [18-Jan-2012].

\begin{thebibliography}{45}

%%% ====================================================================
%%% NOTE TO THE USER: you can override these defaults by providing
%%% customized versions of any of these macros before the \bibliography
%%% command.  Each of them MUST provide its own final punctuation,
%%% except for \shownote{}, \showDOI{}, and \showURL{}.  The latter two
%%% do not use final punctuation, in order to avoid confusing it with
%%% the Web address.
%%%
%%% To suppress output of a particular field, define its macro to expand
%%% to an empty string, or better, \unskip, like this:
%%%
%%% \newcommand{\showDOI}[1]{\unskip}   % LaTeX syntax
%%%
%%% \def \showDOI #1{\unskip}           % plain TeX syntax
%%%
%%% ====================================================================

\ifx \showCODEN    \undefined \def \showCODEN     #1{\unskip}     \fi
\ifx \showDOI      \undefined \def \showDOI       #1{#1}\fi
\ifx \showISBNx    \undefined \def \showISBNx     #1{\unskip}     \fi
\ifx \showISBNxiii \undefined \def \showISBNxiii  #1{\unskip}     \fi
\ifx \showISSN     \undefined \def \showISSN      #1{\unskip}     \fi
\ifx \showLCCN     \undefined \def \showLCCN      #1{\unskip}     \fi
\ifx \shownote     \undefined \def \shownote      #1{#1}          \fi
\ifx \showarticletitle \undefined \def \showarticletitle #1{#1}   \fi
\ifx \showURL      \undefined \def \showURL       {\relax}        \fi
% The following commands are used for tagged output and should be
% invisible to TeX
\providecommand\bibfield[2]{#2}
\providecommand\bibinfo[2]{#2}
\providecommand\natexlab[1]{#1}
\providecommand\showeprint[2][]{arXiv:#2}

\bibitem[\protect\citeauthoryear{Abbasi, Banerjee, Byrka, Chalermsook, Gadekar,
  Khodamoradi, Marx, Sharma, and Spoerhase}{Abbasi et~al\mbox{.}}{2023}]%
        {abbasi2023parameterized}
\bibfield{author}{\bibinfo{person}{Fateme Abbasi}, \bibinfo{person}{Sandip
  Banerjee}, \bibinfo{person}{Jaroslaw Byrka}, \bibinfo{person}{Parinya
  Chalermsook}, \bibinfo{person}{Ameet Gadekar}, \bibinfo{person}{Kamyar
  Khodamoradi}, \bibinfo{person}{D{\'{a}}niel Marx}, \bibinfo{person}{Roohani
  Sharma}, {and} \bibinfo{person}{Joachim Spoerhase}.}
  \bibinfo{year}{2023}\natexlab{}.
\newblock \showarticletitle{Parameterized Approximation Schemes for Clustering
  with General Norm Objectives}. In \bibinfo{booktitle}{\emph{64th {IEEE}
  Annual Symposium on Foundations of Computer Science, {FOCS} 2023, Santa Cruz,
  CA, USA, November 6-9, 2023}}. \bibinfo{publisher}{{IEEE}},
  \bibinfo{pages}{1377--1399}.
\newblock
\urldef\tempurl%
\url{https://doi.org/10.1109/FOCS57990.2023.00085}
\showDOI{\tempurl}


\bibitem[\protect\citeauthoryear{Anagnostides, Fotakis, and
  Patsilinakos}{Anagnostides et~al\mbox{.}}{2022}]%
        {anagnostides2022metric}
\bibfield{author}{\bibinfo{person}{Ioannis Anagnostides},
  \bibinfo{person}{Dimitris Fotakis}, {and} \bibinfo{person}{Panagiotis
  Patsilinakos}.} \bibinfo{year}{2022}\natexlab{}.
\newblock \showarticletitle{Metric-Distortion Bounds under Limited
  Information}.
\newblock \bibinfo{journal}{\emph{J. Artif. Intell. Res.}}
  \bibinfo{volume}{74} (\bibinfo{year}{2022}), \bibinfo{pages}{1449--1483}.
\newblock
\urldef\tempurl%
\url{https://doi.org/10.1613/JAIR.1.13338}
\showDOI{\tempurl}


\bibitem[\protect\citeauthoryear{Anshelevich, Bhardwaj, Elkind, Postl, and
  Skowron}{Anshelevich et~al\mbox{.}}{2018}]%
        {anshelevich2018approximating}
\bibfield{author}{\bibinfo{person}{Elliot Anshelevich}, \bibinfo{person}{Onkar
  Bhardwaj}, \bibinfo{person}{Edith Elkind}, \bibinfo{person}{John Postl},
  {and} \bibinfo{person}{Piotr Skowron}.} \bibinfo{year}{2018}\natexlab{}.
\newblock \showarticletitle{Approximating Optimal Social Choice under Metric
  Preferences}.
\newblock \bibinfo{journal}{\emph{Artif. Intell.}}  \bibinfo{volume}{264}
  (\bibinfo{year}{2018}), \bibinfo{pages}{27--51}.
\newblock
\urldef\tempurl%
\url{https://doi.org/10.1016/J.ARTINT.2018.07.006}
\showDOI{\tempurl}


\bibitem[\protect\citeauthoryear{Bogomolnaia and Laslier}{Bogomolnaia and
  Laslier}{2007}]%
        {bogomolnaia2007euclidean}
\bibfield{author}{\bibinfo{person}{Anna Bogomolnaia} {and}
  \bibinfo{person}{Jean-Fran{\c{c}}ois Laslier}.}
  \bibinfo{year}{2007}\natexlab{}.
\newblock \showarticletitle{Euclidean Preferences}.
\newblock \bibinfo{journal}{\emph{Journal of Mathematical Economics}}
  \bibinfo{volume}{43}, \bibinfo{number}{2} (\bibinfo{year}{2007}),
  \bibinfo{pages}{87--98}.
\newblock
\urldef\tempurl%
\url{https://doi.org/10.1016/j.jmateco.2006.09.004}
\showDOI{\tempurl}


\bibitem[\protect\citeauthoryear{Caragiannis, Shah, and Voudouris}{Caragiannis
  et~al\mbox{.}}{2022}]%
        {caragiannis2022metric}
\bibfield{author}{\bibinfo{person}{Ioannis Caragiannis},
  \bibinfo{person}{Nisarg Shah}, {and} \bibinfo{person}{Alexandros~A.
  Voudouris}.} \bibinfo{year}{2022}\natexlab{}.
\newblock \showarticletitle{The Metric Distortion of Multiwinner Voting}.
\newblock \bibinfo{journal}{\emph{Artif. Intell.}}  \bibinfo{volume}{313}
  (\bibinfo{year}{2022}), \bibinfo{pages}{103802}.
\newblock
\urldef\tempurl%
\url{https://doi.org/10.1016/J.ARTINT.2022.103802}
\showDOI{\tempurl}


\bibitem[\protect\citeauthoryear{Carroll, Lewis, Lo, Poole, and
  Rosenthal}{Carroll et~al\mbox{.}}{2013}]%
        {carroll2013structure}
\bibfield{author}{\bibinfo{person}{Royce Carroll}, \bibinfo{person}{Jeffrey~B
  Lewis}, \bibinfo{person}{James Lo}, \bibinfo{person}{Keith~T Poole}, {and}
  \bibinfo{person}{Howard Rosenthal}.} \bibinfo{year}{2013}\natexlab{}.
\newblock \showarticletitle{The Structure of Utility in Spatial Models of
  Voting}.
\newblock \bibinfo{journal}{\emph{American Journal of Political Science}}
  \bibinfo{volume}{57}, \bibinfo{number}{4} (\bibinfo{year}{2013}),
  \bibinfo{pages}{1008--1028}.
\newblock
\urldef\tempurl%
\url{https://doi.org/10.1111/ajps.12029}
\showDOI{\tempurl}


\bibitem[\protect\citeauthoryear{Chambers and Echenique}{Chambers and
  Echenique}{2020}]%
        {chambers2020spherical}
\bibfield{author}{\bibinfo{person}{Christopher~P Chambers} {and}
  \bibinfo{person}{Federico Echenique}.} \bibinfo{year}{2020}\natexlab{}.
\newblock \showarticletitle{Spherical Preferences}.
\newblock \bibinfo{journal}{\emph{Journal of Economic Theory}}
  \bibinfo{volume}{189} (\bibinfo{year}{2020}), \bibinfo{pages}{105086}.
\newblock
\urldef\tempurl%
\url{https://doi.org/10.1016/j.jet.2020.105086}
\showDOI{\tempurl}


\bibitem[\protect\citeauthoryear{Chan, Filos{-}Ratsikas, Li, Li, and Wang}{Chan
  et~al\mbox{.}}{2021}]%
        {chan2021mechanism}
\bibfield{author}{\bibinfo{person}{Hau Chan}, \bibinfo{person}{Aris
  Filos{-}Ratsikas}, \bibinfo{person}{Bo Li}, \bibinfo{person}{Minming Li},
  {and} \bibinfo{person}{Chenhao Wang}.} \bibinfo{year}{2021}\natexlab{}.
\newblock \showarticletitle{Mechanism Design for Facility Location Problems:
  {A} Survey}. In \bibinfo{booktitle}{\emph{Proceedings of the Thirtieth
  International Joint Conference on Artificial Intelligence, {IJCAI} 2021,
  Virtual Event / Montreal, Canada, 19-27 August 2021}}.
  \bibinfo{publisher}{ijcai.org}, \bibinfo{pages}{4356--4365}.
\newblock
\urldef\tempurl%
\url{https://doi.org/10.24963/IJCAI.2021/596}
\showDOI{\tempurl}


\bibitem[\protect\citeauthoryear{Chen, N{\"{o}}llenburg, Simola, Villedieu, and
  Wallinger}{Chen et~al\mbox{.}}{2022}]%
        {chen2022manhattan}
\bibfield{author}{\bibinfo{person}{Jiehua Chen}, \bibinfo{person}{Martin
  N{\"{o}}llenburg}, \bibinfo{person}{Sofia Simola},
  \bibinfo{person}{Ana{\"{\i}}s Villedieu}, {and} \bibinfo{person}{Markus
  Wallinger}.} \bibinfo{year}{2022}\natexlab{}.
\newblock \showarticletitle{Multidimensional Manhattan Preferences}. In
  \bibinfo{booktitle}{\emph{{LATIN} 2022: Theoretical Informatics - 15th Latin
  American Symposium, Guanajuato, Mexico, November 7-11, 2022, Proceedings}}
  \emph{(\bibinfo{series}{Lecture Notes in Computer Science},
  Vol.~\bibinfo{volume}{13568})}. \bibinfo{publisher}{Springer},
  \bibinfo{pages}{273--289}.
\newblock
\urldef\tempurl%
\url{https://doi.org/10.1007/978-3-031-20624-5_17}
\showDOI{\tempurl}


\bibitem[\protect\citeauthoryear{Clinton, Jackman, and Rivers}{Clinton
  et~al\mbox{.}}{2004}]%
        {clinton2004statistical}
\bibfield{author}{\bibinfo{person}{Joshua Clinton}, \bibinfo{person}{Simon
  Jackman}, {and} \bibinfo{person}{Douglas Rivers}.}
  \bibinfo{year}{2004}\natexlab{}.
\newblock \showarticletitle{The Statistical Analysis of Roll Call Data}.
\newblock \bibinfo{journal}{\emph{American Political Science Review}}
  \bibinfo{volume}{98}, \bibinfo{number}{2} (\bibinfo{year}{2004}),
  \bibinfo{pages}{355--370}.
\newblock
\urldef\tempurl%
\url{https://doi.org/10.1017/S0003055404001194}
\showDOI{\tempurl}


\bibitem[\protect\citeauthoryear{Cohen{-}Addad, Grandoni, Lee, and
  Schwiegelshohn}{Cohen{-}Addad et~al\mbox{.}}{2023}]%
        {cohen2023breaching}
\bibfield{author}{\bibinfo{person}{Vincent Cohen{-}Addad},
  \bibinfo{person}{Fabrizio Grandoni}, \bibinfo{person}{Euiwoong Lee}, {and}
  \bibinfo{person}{Chris Schwiegelshohn}.} \bibinfo{year}{2023}\natexlab{}.
\newblock \showarticletitle{Breaching the 2 {LMP} Approximation Barrier for
  Facility Location with Applications to \emph{k}-Median}. In
  \bibinfo{booktitle}{\emph{Proceedings of the 2023 {ACM-SIAM} Symposium on
  Discrete Algorithms, {SODA} 2023, Florence, Italy, January 22-25, 2023}}.
  \bibinfo{publisher}{{SIAM}}, \bibinfo{pages}{940--986}.
\newblock
\urldef\tempurl%
\url{https://doi.org/10.1137/1.9781611977554.CH37}
\showDOI{\tempurl}


\bibitem[\protect\citeauthoryear{Cohen{-}Addad and {Karthik {C.
  S.}}}{Cohen{-}Addad and {Karthik {C. S.}}}{2019}]%
        {cohen2019inapproximability}
\bibfield{author}{\bibinfo{person}{Vincent Cohen{-}Addad} {and}
  \bibinfo{person}{{Karthik {C. S.}}}} \bibinfo{year}{2019}\natexlab{}.
\newblock \showarticletitle{Inapproximability of Clustering in Lp Metrics}. In
  \bibinfo{booktitle}{\emph{60th {IEEE} Annual Symposium on Foundations of
  Computer Science, {FOCS} 2019, Baltimore, Maryland, USA, November 9-12,
  2019}}. \bibinfo{publisher}{{IEEE} Computer Society},
  \bibinfo{pages}{519--539}.
\newblock
\urldef\tempurl%
\url{https://doi.org/10.1109/FOCS.2019.00040}
\showDOI{\tempurl}


\bibitem[\protect\citeauthoryear{Cohen{-}Addad, {Karthik {C. S.}}, and
  Lee}{Cohen{-}Addad et~al\mbox{.}}{2022}]%
        {cohen2022johnson}
\bibfield{author}{\bibinfo{person}{Vincent Cohen{-}Addad},
  \bibinfo{person}{{Karthik {C. S.}}}, {and} \bibinfo{person}{Euiwoong Lee}.}
  \bibinfo{year}{2022}\natexlab{}.
\newblock \showarticletitle{Johnson Coverage Hypothesis: Inapproximability of
  \emph{k}-Means and \emph{k}-Median in \emph{p}-Metrics}. In
  \bibinfo{booktitle}{\emph{Proceedings of the 2022 {ACM-SIAM} Symposium on
  Discrete Algorithms, {SODA} 2022, Virtual Conference / Alexandria, VA, USA,
  January 9-12, 2022}}. \bibinfo{publisher}{{SIAM}},
  \bibinfo{pages}{1493--1530}.
\newblock
\urldef\tempurl%
\url{https://doi.org/10.1137/1.9781611977073.63}
\showDOI{\tempurl}


\bibitem[\protect\citeauthoryear{Ebadian, Halpern, and Micha}{Ebadian
  et~al\mbox{.}}{2024}]%
        {ebadian2024metric}
\bibfield{author}{\bibinfo{person}{Soroush Ebadian}, \bibinfo{person}{Daniel
  Halpern}, {and} \bibinfo{person}{Evi Micha}.}
  \bibinfo{year}{2024}\natexlab{}.
\newblock \showarticletitle{Metric Distortion with Elicited Pairwise
  Comparisons}. In \bibinfo{booktitle}{\emph{Proceedings of the Thirty-Third
  International Joint Conference on Artificial Intelligence, {IJCAI} 2024,
  Jeju, South Korea, August 3-9, 2024}}. \bibinfo{publisher}{ijcai.org},
  \bibinfo{pages}{2791--2798}.
\newblock
\urldef\tempurl%
\url{https://doi.org/10.24963/ijcai.2024/309}
\showDOI{\tempurl}


\bibitem[\protect\citeauthoryear{Eckert and Klamler}{Eckert and
  Klamler}{2010}]%
        {eckert2010equity}
\bibfield{author}{\bibinfo{person}{Daniel Eckert} {and}
  \bibinfo{person}{Christian Klamler}.} \bibinfo{year}{2010}\natexlab{}.
\newblock \showarticletitle{An Equity-Efficiency Trade-Off in a Geometric
  Approach to Committee Selection}.
\newblock \bibinfo{journal}{\emph{European Journal of Political Economy}}
  \bibinfo{volume}{26}, \bibinfo{number}{3} (\bibinfo{year}{2010}),
  \bibinfo{pages}{386--391}.
\newblock
\urldef\tempurl%
\url{https://doi.org/10.1016/j.ejpoleco.2009.11.009}
\showDOI{\tempurl}


\bibitem[\protect\citeauthoryear{Eguia}{Eguia}{2011}]%
        {eguia2011foundations}
\bibfield{author}{\bibinfo{person}{Jon~X Eguia}.}
  \bibinfo{year}{2011}\natexlab{}.
\newblock \showarticletitle{Foundations of Spatial Preferences}.
\newblock \bibinfo{journal}{\emph{Journal of Mathematical Economics}}
  \bibinfo{volume}{47}, \bibinfo{number}{2} (\bibinfo{year}{2011}),
  \bibinfo{pages}{200--205}.
\newblock
\urldef\tempurl%
\url{https://doi.org/10.1016/j.jmateco.2010.12.014}
\showDOI{\tempurl}


\bibitem[\protect\citeauthoryear{Eguia}{Eguia}{2013}]%
        {eguia2013spatial}
\bibfield{author}{\bibinfo{person}{Jon~X Eguia}.}
  \bibinfo{year}{2013}\natexlab{}.
\newblock \showarticletitle{On the Spatial Representation of Preference
  Profiles}.
\newblock \bibinfo{journal}{\emph{Economic Theory}} \bibinfo{volume}{52},
  \bibinfo{number}{1} (\bibinfo{year}{2013}), \bibinfo{pages}{103--128}.
\newblock
\urldef\tempurl%
\url{https://doi.org/10.1007/s00199-011-0669-8}
\showDOI{\tempurl}


\bibitem[\protect\citeauthoryear{Eguia et~al\mbox{.}}{Eguia
  et~al\mbox{.}}{2009}]%
        {eguia2009utility}
\bibfield{author}{\bibinfo{person}{Jon~X Eguia} {et~al\mbox{.}}}
  \bibinfo{year}{2009}\natexlab{}.
\newblock \showarticletitle{Utility Representations of Risk Neutral Preferences
  in Multiple Dimensions}.
\newblock \bibinfo{journal}{\emph{Quarterly Journal of Political Science}}
  \bibinfo{volume}{4}, \bibinfo{number}{4} (\bibinfo{year}{2009}),
  \bibinfo{pages}{379--385}.
\newblock
\urldef\tempurl%
\url{https://doi.org/10.1561/100.00009059}
\showDOI{\tempurl}


\bibitem[\protect\citeauthoryear{Enelow and Hinich}{Enelow and Hinich}{1984}]%
        {enelow1984spatial}
\bibfield{author}{\bibinfo{person}{James~M Enelow} {and}
  \bibinfo{person}{Melvin~J Hinich}.} \bibinfo{year}{1984}\natexlab{}.
\newblock \bibinfo{booktitle}{\emph{The Spatial Theory of Voting: An
  Introduction}}.
\newblock \bibinfo{publisher}{Cambridge University Press}.
\newblock


\bibitem[\protect\citeauthoryear{Escoffier, Spanjaard, and
  Tydrichov{\'a}}{Escoffier et~al\mbox{.}}{2023}]%
        {escoffier2023algorithmic}
\bibfield{author}{\bibinfo{person}{Bruno Escoffier}, \bibinfo{person}{Olivier
  Spanjaard}, {and} \bibinfo{person}{Magdal{\'e}na Tydrichov{\'a}}.}
  \bibinfo{year}{2023}\natexlab{}.
\newblock \showarticletitle{Algorithmic Recognition of 2-Euclidean
  Preferences}.
\newblock In \bibinfo{booktitle}{\emph{ECAI 2023}}. \bibinfo{publisher}{IOS
  Press}, \bibinfo{pages}{637--644}.
\newblock
\urldef\tempurl%
\url{https://doi.org/10.3233/FAIA230326}
\showDOI{\tempurl}


\bibitem[\protect\citeauthoryear{Escoffier, Spanjaard, and
  Tydrichov{\'{a}}}{Escoffier et~al\mbox{.}}{2024}]%
        {escoffier2024euclidean}
\bibfield{author}{\bibinfo{person}{Bruno Escoffier}, \bibinfo{person}{Olivier
  Spanjaard}, {and} \bibinfo{person}{Magdal{\'{e}}na Tydrichov{\'{a}}}.}
  \bibinfo{year}{2024}\natexlab{}.
\newblock \showarticletitle{Euclidean Preferences in the Plane under
  {$\ell_1$}, {$\ell_2$} and {$\ell_{\infty}$} Norms}.
\newblock \bibinfo{journal}{\emph{Soc. Choice Welf.}} \bibinfo{volume}{63},
  \bibinfo{number}{1} (\bibinfo{year}{2024}), \bibinfo{pages}{125--169}.
\newblock
\urldef\tempurl%
\url{https://doi.org/10.1007/S00355-024-01525-2}
\showDOI{\tempurl}


\bibitem[\protect\citeauthoryear{Feigenbaum, Sethuraman, and Ye}{Feigenbaum
  et~al\mbox{.}}{2017}]%
        {feigenbaum2017approximately}
\bibfield{author}{\bibinfo{person}{Itai Feigenbaum}, \bibinfo{person}{Jay
  Sethuraman}, {and} \bibinfo{person}{Chun Ye}.}
  \bibinfo{year}{2017}\natexlab{}.
\newblock \showarticletitle{Approximately Optimal Mechanisms for Strategyproof
  Facility Location: Minimizing lp Norm of Costs}.
\newblock \bibinfo{journal}{\emph{Mathematics of Operations Research}}
  \bibinfo{volume}{42}, \bibinfo{number}{2} (\bibinfo{year}{2017}),
  \bibinfo{pages}{434--447}.
\newblock
\urldef\tempurl%
\url{https://doi.org/10.1287/moor.2016.0810}
\showDOI{\tempurl}


\bibitem[\protect\citeauthoryear{Goodfellow, McDaniel, and Papernot}{Goodfellow
  et~al\mbox{.}}{2018}]%
        {goodfellow2018machine}
\bibfield{author}{\bibinfo{person}{Ian~J. Goodfellow},
  \bibinfo{person}{Patrick~D. McDaniel}, {and} \bibinfo{person}{Nicolas
  Papernot}.} \bibinfo{year}{2018}\natexlab{}.
\newblock \showarticletitle{Making Machine Learning Robust against Adversarial
  Inputs}.
\newblock \bibinfo{journal}{\emph{Commun. {ACM}}} \bibinfo{volume}{61},
  \bibinfo{number}{7} (\bibinfo{year}{2018}), \bibinfo{pages}{56--66}.
\newblock
\urldef\tempurl%
\url{https://doi.org/10.1145/3134599}
\showDOI{\tempurl}


\bibitem[\protect\citeauthoryear{Henry and Mourifi{\'e}}{Henry and
  Mourifi{\'e}}{2013}]%
        {henry2013euclidean}
\bibfield{author}{\bibinfo{person}{Marc Henry} {and} \bibinfo{person}{Ismael
  Mourifi{\'e}}.} \bibinfo{year}{2013}\natexlab{}.
\newblock \showarticletitle{Euclidean Revealed Preferences: Testing the Spatial
  Voting Model}.
\newblock \bibinfo{journal}{\emph{Journal of Applied Econometrics}}
  \bibinfo{volume}{28}, \bibinfo{number}{4} (\bibinfo{year}{2013}),
  \bibinfo{pages}{650--666}.
\newblock
\urldef\tempurl%
\url{https://doi.org/10.1002/jae.1276}
\showDOI{\tempurl}


\bibitem[\protect\citeauthoryear{Humphreys and Laver}{Humphreys and
  Laver}{2010}]%
        {humphreys2010spatial}
\bibfield{author}{\bibinfo{person}{Macartan Humphreys} {and}
  \bibinfo{person}{Michael Laver}.} \bibinfo{year}{2010}\natexlab{}.
\newblock \showarticletitle{Spatial Models, Cognitive Metrics, and Majority
  Rule Equilibria}.
\newblock \bibinfo{journal}{\emph{British Journal of Political Science}}
  \bibinfo{volume}{40}, \bibinfo{number}{1} (\bibinfo{year}{2010}),
  \bibinfo{pages}{11--30}.
\newblock
\urldef\tempurl%
\url{https://doi.org/10.1017/S0007123409990263}
\showDOI{\tempurl}


\bibitem[\protect\citeauthoryear{Kalandrakis}{Kalandrakis}{2010}]%
        {kalandrakis2010rationalizable}
\bibfield{author}{\bibinfo{person}{Tasos Kalandrakis}.}
  \bibinfo{year}{2010}\natexlab{}.
\newblock \showarticletitle{Rationalizable Voting}.
\newblock \bibinfo{journal}{\emph{Theoretical Economics}} \bibinfo{volume}{5},
  \bibinfo{number}{1} (\bibinfo{year}{2010}), \bibinfo{pages}{93--125}.
\newblock
\urldef\tempurl%
\url{https://doi.org/10.3982/TE425}
\showDOI{\tempurl}


\bibitem[\protect\citeauthoryear{Kim, Londregan, and Ratkovic}{Kim
  et~al\mbox{.}}{2018}]%
        {kim2018estimating}
\bibfield{author}{\bibinfo{person}{In~Song Kim}, \bibinfo{person}{John
  Londregan}, {and} \bibinfo{person}{Marc Ratkovic}.}
  \bibinfo{year}{2018}\natexlab{}.
\newblock \showarticletitle{Estimating Spatial Preferences from Votes and
  Text}.
\newblock \bibinfo{journal}{\emph{Political Analysis}} \bibinfo{volume}{26},
  \bibinfo{number}{2} (\bibinfo{year}{2018}), \bibinfo{pages}{210--229}.
\newblock
\urldef\tempurl%
\url{https://doi.org/10.1017/pan.2018.7}
\showDOI{\tempurl}


\bibitem[\protect\citeauthoryear{Knoblauch}{Knoblauch}{2010}]%
        {knoblauch2010recognizing}
\bibfield{author}{\bibinfo{person}{Vicki Knoblauch}.}
  \bibinfo{year}{2010}\natexlab{}.
\newblock \showarticletitle{Recognizing One-Dimensional Euclidean Preference
  Profiles}.
\newblock \bibinfo{journal}{\emph{Journal of Mathematical Economics}}
  \bibinfo{volume}{46}, \bibinfo{number}{1} (\bibinfo{year}{2010}),
  \bibinfo{pages}{1--5}.
\newblock
\urldef\tempurl%
\url{https://doi.org/10.1016/j.jmateco.2009.05.007}
\showDOI{\tempurl}


\bibitem[\protect\citeauthoryear{Ladha}{Ladha}{1991}]%
        {ladha1991spatial}
\bibfield{author}{\bibinfo{person}{Krishna~K Ladha}.}
  \bibinfo{year}{1991}\natexlab{}.
\newblock \showarticletitle{A Spatial Model of Legislative Voting with
  Perceptual Error}.
\newblock \bibinfo{journal}{\emph{Public Choice}} \bibinfo{volume}{68},
  \bibinfo{number}{1} (\bibinfo{year}{1991}), \bibinfo{pages}{151--174}.
\newblock
\urldef\tempurl%
\url{https://doi.org/10.1007/BF00173825}
\showDOI{\tempurl}


\bibitem[\protect\citeauthoryear{Larson and Sadiq}{Larson and Sadiq}{1983}]%
        {larson1983facility}
\bibfield{author}{\bibinfo{person}{Richard~C. Larson} {and}
  \bibinfo{person}{Ghazala Sadiq}.} \bibinfo{year}{1983}\natexlab{}.
\newblock \showarticletitle{Facility Locations with the Manhattan Metric in the
  Presence of Barriers to Travel}.
\newblock \bibinfo{journal}{\emph{Oper. Res.}} \bibinfo{volume}{31},
  \bibinfo{number}{4} (\bibinfo{year}{1983}), \bibinfo{pages}{652--669}.
\newblock
\urldef\tempurl%
\url{https://doi.org/10.1287/OPRE.31.4.652}
\showDOI{\tempurl}


\bibitem[\protect\citeauthoryear{Lee and Shin}{Lee and Shin}{2025}]%
        {lee2025facility}
\bibfield{author}{\bibinfo{person}{Euiwoong Lee} {and} \bibinfo{person}{Kijun
  Shin}.} \bibinfo{year}{2025}\natexlab{}.
\newblock \showarticletitle{Facility Location on High-Dimensional Euclidean
  Spaces}. In \bibinfo{booktitle}{\emph{16th Innovations in Theoretical
  Computer Science Conference, {ITCS} 2025, January 7-10, 2025, New York, NY,
  {USA}}} \emph{(\bibinfo{series}{LIPIcs}, Vol.~\bibinfo{volume}{325})}.
  \bibinfo{publisher}{Schloss Dagstuhl - Leibniz-Zentrum f{\"{u}}r Informatik},
  \bibinfo{pages}{70:1--70:21}.
\newblock
\urldef\tempurl%
\url{https://doi.org/10.4230/LIPICS.ITCS.2025.70}
\showDOI{\tempurl}


\bibitem[\protect\citeauthoryear{Luque and Sosa}{Luque and Sosa}{2025}]%
        {luque2025operationalizing}
\bibfield{author}{\bibinfo{person}{Carolina Luque} {and} \bibinfo{person}{Juan
  Sosa}.} \bibinfo{year}{2025}\natexlab{}.
\newblock \showarticletitle{Operationalizing Legislative Bodies: A
  Methodological and Empirical Perspective with a Bayesian Approach}.
\newblock \bibinfo{journal}{\emph{J. Amer. Statist. Assoc.}}
  \bibinfo{volume}{120}, \bibinfo{number}{549} (\bibinfo{year}{2025}),
  \bibinfo{pages}{572--583}.
\newblock
\urldef\tempurl%
\url{https://doi.org/10.1080/01621459.2024.2413928}
\showDOI{\tempurl}


\bibitem[\protect\citeauthoryear{Milyo}{Milyo}{2000}]%
        {milyo2000logical}
\bibfield{author}{\bibinfo{person}{Jeffrey Milyo}.}
  \bibinfo{year}{2000}\natexlab{}.
\newblock \showarticletitle{Logical Deficiencies in Spatial Models: A
  Constructive Critique}.
\newblock \bibinfo{journal}{\emph{Public Choice}} \bibinfo{volume}{105},
  \bibinfo{number}{3} (\bibinfo{year}{2000}), \bibinfo{pages}{273--289}.
\newblock
\urldef\tempurl%
\url{https://doi.org/10.1023/A:1005203611524}
\showDOI{\tempurl}


\bibitem[\protect\citeauthoryear{Pardos-Prado and Dinas}{Pardos-Prado and
  Dinas}{2010}]%
        {pardos2010systemic}
\bibfield{author}{\bibinfo{person}{Sergi Pardos-Prado} {and}
  \bibinfo{person}{Elias Dinas}.} \bibinfo{year}{2010}\natexlab{}.
\newblock \showarticletitle{Systemic Polarisation and Spatial Voting}.
\newblock \bibinfo{journal}{\emph{European Journal of Political Research}}
  \bibinfo{volume}{49}, \bibinfo{number}{6} (\bibinfo{year}{2010}),
  \bibinfo{pages}{759--786}.
\newblock
\urldef\tempurl%
\url{https://doi.org/10.1111/j.1475-6765.2010.01918.x}
\showDOI{\tempurl}


\bibitem[\protect\citeauthoryear{Peters}{Peters}{2017}]%
        {peters2017recognising}
\bibfield{author}{\bibinfo{person}{Dominik Peters}.}
  \bibinfo{year}{2017}\natexlab{}.
\newblock \showarticletitle{Recognising Multidimensional Euclidean
  Preferences}. In \bibinfo{booktitle}{\emph{Proceedings of the Thirty-First
  {AAAI} Conference on Artificial Intelligence, February 4-9, 2017, San
  Francisco, California, {USA}}}. \bibinfo{publisher}{{AAAI} Press},
  \bibinfo{pages}{642--648}.
\newblock
\urldef\tempurl%
\url{https://doi.org/10.1609/AAAI.V31I1.10616}
\showDOI{\tempurl}


\bibitem[\protect\citeauthoryear{Peters, van~der Stel, and Storcken}{Peters
  et~al\mbox{.}}{1993}]%
        {peters1993generalized}
\bibfield{author}{\bibinfo{person}{Hans Peters}, \bibinfo{person}{Hans van~der
  Stel}, {and} \bibinfo{person}{Ton Storcken}.}
  \bibinfo{year}{1993}\natexlab{}.
\newblock \showarticletitle{Generalized Median Solutions, Strategy-Proofness
  and Strictly Convex Norms}.
\newblock \bibinfo{journal}{\emph{Zeitschrift f{\"u}r Operations Research}}
  \bibinfo{volume}{38}, \bibinfo{number}{1} (\bibinfo{year}{1993}),
  \bibinfo{pages}{19--53}.
\newblock
\urldef\tempurl%
\url{https://doi.org/10.1007/BF01416005}
\showDOI{\tempurl}


\bibitem[\protect\citeauthoryear{Rudin}{Rudin}{2009}]%
        {rudin2009pnorm}
\bibfield{author}{\bibinfo{person}{Cynthia Rudin}.}
  \bibinfo{year}{2009}\natexlab{}.
\newblock \showarticletitle{The P-Norm Push: {A} Simple Convex Ranking
  Algorithm that Concentrates at the Top of the List}.
\newblock \bibinfo{journal}{\emph{J. Mach. Learn. Res.}}  \bibinfo{volume}{10}
  (\bibinfo{year}{2009}), \bibinfo{pages}{2233--2271}.
\newblock
\urldef\tempurl%
\url{https://doi.org/10.5555/1577069.1755861}
\showDOI{\tempurl}


\bibitem[\protect\citeauthoryear{Shin, Lim, and Park}{Shin
  et~al\mbox{.}}{2025}]%
        {shin2025l1}
\bibfield{author}{\bibinfo{person}{Sooahn Shin}, \bibinfo{person}{Johan Lim},
  {and} \bibinfo{person}{Jong~Hee Park}.} \bibinfo{year}{2025}\natexlab{}.
\newblock \showarticletitle{{$\ell^1$}-based Bayesian Ideal Point Model for
  Multidimensional Politics}.
\newblock \bibinfo{journal}{\emph{J. Amer. Statist. Assoc.}}
  \bibinfo{volume}{120}, \bibinfo{number}{550} (\bibinfo{year}{2025}),
  \bibinfo{pages}{631--644}.
\newblock
\urldef\tempurl%
\url{https://doi.org/10.1080/01621459.2024.2425461}
\showDOI{\tempurl}


\bibitem[\protect\citeauthoryear{Stoetzer and Zittlau}{Stoetzer and
  Zittlau}{2015}]%
        {stoetzer2015multidimensional}
\bibfield{author}{\bibinfo{person}{Lukas~F Stoetzer} {and}
  \bibinfo{person}{Steffen Zittlau}.} \bibinfo{year}{2015}\natexlab{}.
\newblock \showarticletitle{Multidimensional Spatial Voting with Non-Separable
  Preferences}.
\newblock \bibinfo{journal}{\emph{Political Analysis}} \bibinfo{volume}{23},
  \bibinfo{number}{3} (\bibinfo{year}{2015}), \bibinfo{pages}{415--428}.
\newblock
\urldef\tempurl%
\url{https://doi.org/10.1093/pan/mpv013}
\showDOI{\tempurl}


\bibitem[\protect\citeauthoryear{Stokes}{Stokes}{1963}]%
        {stokes1963spatial}
\bibfield{author}{\bibinfo{person}{Donald~E Stokes}.}
  \bibinfo{year}{1963}\natexlab{}.
\newblock \showarticletitle{Spatial Models of Party Competition}.
\newblock \bibinfo{journal}{\emph{American Political Science Review}}
  \bibinfo{volume}{57}, \bibinfo{number}{2} (\bibinfo{year}{1963}),
  \bibinfo{pages}{368--377}.
\newblock
\urldef\tempurl%
\url{https://doi.org/10.2307/1952828}
\showDOI{\tempurl}


\bibitem[\protect\citeauthoryear{Thorburn, Polukarov, and Ventre}{Thorburn
  et~al\mbox{.}}{2023}]%
        {thorburn2023error}
\bibfield{author}{\bibinfo{person}{Luke Thorburn}, \bibinfo{person}{Maria
  Polukarov}, {and} \bibinfo{person}{Carmine Ventre}.}
  \bibinfo{year}{2023}\natexlab{}.
\newblock \showarticletitle{Error in the Euclidean Preference Model}. In
  \bibinfo{booktitle}{\emph{Proceedings of the Thirty-Second International
  Joint Conference on Artificial Intelligence}} (Macao, P.R.China)
  \emph{(\bibinfo{series}{IJCAI '23})}. Article \bibinfo{articleno}{322},
  \bibinfo{numpages}{9}~pages.
\newblock
\urldef\tempurl%
\url{https://doi.org/10.24963/ijcai.2023/322}
\showDOI{\tempurl}


\bibitem[\protect\citeauthoryear{Totik}{Totik}{2013}]%
        {totik2013approximation}
\bibfield{author}{\bibinfo{person}{Vilmos Totik}.}
  \bibinfo{year}{2013}\natexlab{}.
\newblock \showarticletitle{Approximation by Homogeneous Polynomials}.
\newblock \bibinfo{journal}{\emph{J. Approx. Theory}}  \bibinfo{volume}{174}
  (\bibinfo{year}{2013}), \bibinfo{pages}{192--205}.
\newblock
\urldef\tempurl%
\url{https://doi.org/10.1016/J.JAT.2013.07.005}
\showDOI{\tempurl}


\bibitem[\protect\citeauthoryear{Ye, Li, and Leiker}{Ye et~al\mbox{.}}{2011}]%
        {ye2011evaluating}
\bibfield{author}{\bibinfo{person}{Min Ye}, \bibinfo{person}{Quan Li}, {and}
  \bibinfo{person}{Kyle~W Leiker}.} \bibinfo{year}{2011}\natexlab{}.
\newblock \showarticletitle{Evaluating Voter--Candidate Proximity in a
  Non-Euclidean Space}.
\newblock \bibinfo{journal}{\emph{Journal of Elections, Public Opinion \&
  Parties}} \bibinfo{volume}{21}, \bibinfo{number}{4} (\bibinfo{year}{2011}),
  \bibinfo{pages}{497--521}.
\newblock
\urldef\tempurl%
\url{https://doi.org/10.1080/17457289.2011.609619}
\showDOI{\tempurl}


\bibitem[\protect\citeauthoryear{Yu and Rodriguez}{Yu and Rodriguez}{2021}]%
        {yu2021spatial}
\bibfield{author}{\bibinfo{person}{Xingchen Yu} {and} \bibinfo{person}{Abel
  Rodriguez}.} \bibinfo{year}{2021}\natexlab{}.
\newblock \showarticletitle{Spatial Voting Models in Circular Spaces: A Case
  Study of the US House of Representatives}.
\newblock \bibinfo{journal}{\emph{The Annals of Applied Statistics}}
  \bibinfo{volume}{15}, \bibinfo{number}{4} (\bibinfo{year}{2021}),
  \bibinfo{pages}{1897--1922}.
\newblock
\urldef\tempurl%
\url{https://doi.org/10.1214/21-AOAS1454}
\showDOI{\tempurl}


\bibitem[\protect\citeauthoryear{Zeitlin and Coupette}{Zeitlin and
  Coupette}{2026}]%
        {zeitlin2026real}
\bibfield{author}{\bibinfo{person}{Joshua Zeitlin} {and}
  \bibinfo{person}{Corinna Coupette}.} \bibinfo{year}{2026}\natexlab{}.
\newblock \showarticletitle{Real Preferences under Arbitrary Norms: Extended
  Abstract}. In \bibinfo{booktitle}{\emph{Proceedings of the 25th International
  Conference on Autonomous Agents and Multiagent Systems, {AAMAS} 2026, Paphos,
  Cyprus, May 25-29, 2026}}. \bibinfo{pages}{3}.
\newblock
\urldef\tempurl%
\url{https://doi.org/10.65109/JWXO7778}
\showDOI{\tempurl}


\end{thebibliography}

\end{document}